\documentclass[a4,oneside,english,reqno,12pt]{amsart}
\usepackage{extsizes}
\usepackage{geometry}
\usepackage{amssymb,amsmath,amsthm,amsfonts,mathtools,esint,dsfont,bbm,yhmath,breqn}
\usepackage[nospace]{cite}
\usepackage{babel}
\usepackage{enumitem}
\usepackage{appendix}
\DeclareMathOperator{\tr}{Tr}
\DeclareMathOperator{\curl}{curl}
\usepackage[pdfencoding=auto,psdextra]{hyperref}
\usepackage{bookmark}
\usepackage{cleveref}
\usepackage{crossreftools}
\crefname{subsection}{subsection}{subsections}
\usepackage{color}
\makeatletter
\numberwithin{equation}{section}
\numberwithin{figure}{section}
\theoremstyle{plain}
\newtheorem{proposition}{Proposition}
\newtheorem{lemma}{Lemma}

\newtheorem{theorem}{\protect\theoremname}
\theoremstyle{remark}
\newtheorem{remark}{Remark}
\newlist{casenv}{enumerate}{4}
\setlist[casenv]{leftmargin=*,align=left,widest={iiii}}
\setlist[casenv,1]{label={{\itshape\ \casename} \arabic*.},ref=\arabic*}
\setlist[casenv,2]{label={{\itshape\ \casename} \mathbb Roman*.},ref=\mathbb Roman*}
\setlist[casenv,3]{label={{\itshape\ \casename\ \alph*.}},ref=\alph*}
\setlist[casenv,4]{label={{\itshape\ \casename} \arabic*.},ref=\arabic*}
\makeatother

\providecommand{\casename}{Case}
\providecommand{\theoremname}{Theorem}

\allowdisplaybreaks
\begin{document}

\title[2D attractive almost-bosonic anyon gases]{2D attractive almost-bosonic anyon gases
}

\author[D.-T. Nguyen]{Dinh-Thi Nguyen}
\address{Department of Mathematics, Uppsala University, Box 480, SE-751 06 Uppsala, Sweden} 
\email{\url{dinh\_thi.nguyen@math.uu.se}}

\begin{abstract}
In two-dimensional space, we consider a system of $N$ anyons interacts via a short range attractive two-body interaction. In the stable regime, we derive the average-field-Pauli functional as the mean-field limit of many-body quantum mechanics. Furthermore, we investigate the collapse phenomenon in the collapse regime where the strength of attractions tends to a critical value (defined by the cubic NLS equation) while simultaneously considering the weak field regime where the strength of the self-generated magnetic field tends to zero.
\end{abstract} 

\subjclass[2020]{81V70, 81Q10, 35P15, 46N50} 

\date{\today.}

\keywords{anyon gases, average-field-Pauli, average-field-Hartree, attractive interactions, Chern--Simon--Schr\"odinger, concentration compactness, ground states, magnetic field, many-body, minimizers}

\maketitle

{
  \hypersetup{hidelinks}
  \tableofcontents
}

\section{Introduction and main results}\label{sec:introduction}

Anyons are exotic particles that can exist in two-dimensional systems and exhibit fractional statistics, meaning that their wave functions acquire a nontrivial phase when particles are exchanged. Unlike fermions (with half-integer spin) or bosons (with integer spin), anyons can have fractional spin statistics. The study of anyons, exotic quasi-particles with fractional statistics, has captivated physicists for decades due to their unique properties and potential applications in condensed matter physics and quantum information science. Anyons exhibit intermediate statistics between fermions and bosons, leading to novel phenomena such as fractional quantum Hall effect and topological quantum computation. In recent years, there has been a growing interest in understanding the behavior of anyons in the presence of external potentials.

In this article, we explore the fascinating world of anyons interacting via attractions. Moreover, we consider the weak-field limit corresponding to the regime where the strength of the external magnetic field is small compared to other energy scales in the system. In this limit, the effects of the external magnetic field can be treated perturbatively, allowing for analytical or numerical analysis. We investigate how the interplay between the statistics of anyons and the presence of point-like attractions gives rise to rich and intriguing physical phenomena, offering new avenues for theoretical exploration and experimental realization.

\subsection{The average-field-Pauli theory of interacting ``almost-bosonic'' anyon gases}

Mathematically, in the mean-field regime, ``almost-bosonic'' interacting anyon gas is described by the \emph{``average-field-Pauli''} functional (also known as Chern--Simon--Schr\"odinger in the literature)

\begin{equation}\label{functional:afP}
\mathcal{E}^{\rm afP}_{\beta,g}[u] := \int_{\mathbb R^{2}} \left[
\left|(-{\rm i}\nabla + \beta\mathbf{A}[|u|^{2}])u\right|^{2}
+ V|u|^{2} - \frac{g}{2} |u|^{4} \right].
\end{equation}
Here $u \in H^{1}(\mathbb R^{2})$ is with unit probability density $\int_{\mathbb R^{2}} |u|^{2} = 1$. The magnetic vector potential $\mathbf{A}[\varrho] : \mathbb R^{2} \to \mathbb R^{2}$ generates a magnetic field
\begin{equation} \label{potential:magnetic}
\mathbf{A}[\varrho] := \nabla^{\perp} w_{0} * \varrho
 \quad \text{with} \quad w_{0}(x) := \log |x|,
\end{equation}
so that
\begin{equation} \label{eq:self-field}
\curl \beta\mathbf{A}[|u|^{2}] = \beta \Delta w_{0} * |u|^{2} = 2\pi\beta|u|^{2}.
\end{equation}
Thus $\beta \in \mathbb{R}$ is the strength of the magnetic field 
(total/fractional number of flux units). By conjugation symmetry $u \to \overline{u}$ it is enough to consider $\beta \geq 0$. The parameter $g \in \mathbb R$ is the strength of scalar point interactions. In this article, we are interested in the case of an attraction, i.e., $g>0$. Finally, the system is trapped by an external scalar trapping potential, $V \in L_{\rm loc}^{\infty}(\mathbb R^{2})$ and $V(x) \to +\infty$ as $|x| \to +\infty$. Our main interest in this paper is the behavior of minimizers of the ground-state energy associated with the functional \eqref{functional:afP}, given by
\begin{equation}\label{energy:afP}
E^{\rm afP}_{\beta,g} := 
\inf \left\{ \mathcal{E}^{\rm afP}_{\beta,g}[u] : 
u \in H^{1} \cap L_{V}^{2} (\mathbb R^{2}),
\int_{\mathbb R^{2}} |u|^{2} = 1 \right\}
\end{equation}
where we have defined the space 
\begin{equation}\label{space:external}
L_{V}^{2} := \left\{u\in L^{2}(\mathbb R^{2}): \int_{\mathbb R^{2}} V|u|^{2} < \infty\right\}.
\end{equation}

In the absence of the magnetic field, i.e., $\beta = 0$, \eqref{functional:afP} reduces to the so-called Gross--Pitaeavskii (also known as nonlinear Schr\"odinger) functional which models the semi-classical theory of Bose--Einstein condensations (BECs). Condensates with repulsive interactions (which corresponds to the case $g<0$) have become an important research subject after their first realization in the laboratory \cite{AndEnsMatWieCor-95,CorWie-02,DavMewAndDruDurKurKet-95,Ketterle-02} and it was studied in great mathematical details \cite{LieSeiYng-00,LieSei-02,LieSeiSolYng-05,LewNamRou-14,Lewin-ICMP,LewNamRou-16,NamRouSei-16,Rougerie-EMS,Rougerie-20}. On the other hand, BECs with attractive interactions (which corresponds to the case $g>0$) were also realized experimentally \cite{BraSacTolHul-95,BraSacHul-97,KagMurShl-98,SacStoHul-98} and it displays behaviors differently from the repulsive BECs since the system might collapse when the two-body interaction is too negative. In the pass few years, many mathematical works deal with the collapse phenomenon of the focusing system \cite{Zhang-00,GuoSei-14,LewNamRou-16,LewNamRou-17,NamRou-20,Rougerie-EMS,Nguyen-20,GuoLuoYan-20,DinNguRou-23}.

In the presence of the self-generated magnetic field \eqref{potential:magnetic}, ground states of \eqref{functional:afP} display vortex lattices \cite{CorDubLunRou-19}. This is similar to the triangular Abrikosov lattice of superfluidity in rotating BECs where \eqref{potential:magnetic} was replaced by the constant magnetic field \cite{Gross-61,AftBlaDal-05,Aftalion-06,Cooper-08,Fetter-09,CorPinRouYng-11,NguRou-22}. See also \cite{LewNamRou-16,LewNamRou-17,NamRou-20,LewNamRou-17-proc,GuoLuoYan-20,DinNguRou-23} for related works of rotating Bose gases in the focusing regime. The functional \eqref{functional:afP} is the effective model in the study of anyon gases which interpolates between Bose gases and Fermi gases. For non-interacting anyon gases (which corresponds to the case $g=0$ in \eqref{functional:afP}), the validity of the so-called ``average-field approximation'' $\mathcal{E}^{{\rm afP}}_{\beta,g=0}$ at fixed $\beta$ were studied rigorously in \cite{LunRou-15} (see also \cite{Girardot-20} and \cite{Lundholm-17,Lundholm-24} for general reviews). The non-interacting system displays different phenomenon by varying $\beta$. While it behaves as non-interacting Boses gases in the weak field limit $|\beta| \to 0$ \cite{LunRou-15}, its behavior is given by the Thomas--Fermi-like model in the strong field limit $|\beta| \to \infty$ \cite{LiBhaMur-92,CorLunRou-17} where ground states generate infinitely many vortices (see also \cite{LunSei-18,GirRou-21} for related works). On the other hand, anyon gases with self-interactions lead to collective behavior and emergent phenomena in the system. The derivation of the ``average-field-Pauli'' approximation \eqref{functional:afP} can be obtained by the method of proof in \cite{LunRou-15} and the quantum de Finetti theory \cite{Rougerie-15} as used in \cite{LewNamRou-16,LewNamRou-17,NamRou-20,Rougerie-20,Rougerie-EMS} for interacting Bose gases. In the literature, Sen and Chitra \cite{Sen-91,SenChi-92a,SenChi-92b} studied anyon gases with repulsive interactions. Furthermore, anyon gases with attractive interactions appeared in the study of soliton solutions to the gauged NLS equation in the Chern--Simon theory \cite{JacPi-90a,JacPi-90b} (see \cite{Bogomolny-76} for the original work and \cite{Tarantello-08,Dunne-95} for excellent textbooks). This is an analog to the study of the famous nonlinear Schr\"odinger (NLS) equation in the classical field theory \cite{Cazenave-89,CazEst-88}.

In the case of attractions, i.e., $g>0$ in \eqref{functional:afP}, the system is unstable since there is a balance between the kinetic energy and interacting potential. In other words, the kinetic and quartic terms in \eqref{functional:afP} behave the same after scaling. Hence, there exists a critical interaction strength $g_{*}(\beta)$ above which the system is unbounded from below. Such a critical value is defined naturally by
\begin{align}\label{eq:critical-value}
g_{*}(\beta) := \inf\left\{\frac{\int_{\mathbb R^{2}} \left|(-{\rm i}\nabla + \beta\mathbf{A}[|u|^{2}]) u\right|^{2}}{\frac{1}{2}\int_{\mathbb R^{2}} |u|^4} : u \in H^{1}(\mathbb R^{2};\mathbb C), \int_{\mathbb R^{2}}|u|^{2} = 1\right\}.
\end{align}
Equivalently, $g_{*}(\beta)$ is the optimal constant of the ``magnetic'' Gagliardo--Nirenberg inequality
\begin{equation}\label{ineq:gn-magnetic}
\frac{g_{*}(\beta)}{2}\int_{\mathbb R^{2}}|u|^{4} \leq \int_{\mathbb R^{2}}|u|^{2} \int_{\mathbb R^{2}}\left|\left(-{\rm i}\nabla + \beta \left(\int_{\mathbb R^{2}}|u|^{2}\right)^{-1}\mathbf{A}[|u|^{2}]\right) u\right|^{2}.
\end{equation}
In the case $\beta=0$, \eqref{ineq:gn-magnetic} reduces to the famous Gagliardo--Nirenberg inequality with the optimal constant $g_{*}(0) = \|Q\|_{L^{2}}^{2} \approx 2\pi \times 1.86\ldots$ (see \cite{Weinstein-83}) where $Q$ is the (unique) positive, radially symmetric solution of the 2D cubic NLS equation
\begin{equation}\label{eq:nls}
-\Delta Q - Q^3 + Q = 0.
\end{equation}
On the other hand, we have
\begin{equation}\label{est:critical-value}
g_{*}(\beta) \geq \max\{g_{*}(0),4\pi|\beta|\},\quad \forall \beta \ne 0
\end{equation}
which follows from the diamagnetic inequality \cite[Theorem 7.21]{LieLos-01} and the Bogomol'nyi inequality  \cite{Bogomolny-76} (see also \cite[Lemma 1.4.1]{Fournais-Helffer} and \cite{CorLunRou-17})
\begin{equation}\label{ineq:curl}
\int_{\mathbb R^{2}} \left|(-{\rm i}\nabla + \beta\mathbf{A}[|u|^{2}]) u\right|^{2} \geq \pm \int_{\mathbb R^{2}} \curl{\beta\mathbf{A}}[|u|^{2}] = 2\pi|\beta|\int_{\mathbb R^{2}}|u|^4 ,\quad \forall u \in H^{1}(\mathbb R^{2};\mathbb C).
\end{equation}
Thus, for large $\beta$, the critical value $g_{*}(\beta)$ is at least $4\pi\beta$. In fact, it was proved in \cite{AtaLunNgu-24} that $g_{*}(\beta) = 4\pi\beta$ for every $\beta \geq 2$ and that the infimum in \eqref{eq:critical-value} is actually attained for $\beta \in 2\mathbb{N}^{*}$. On the other hand, for small $\beta$, which is of our main interest in this article, we have that $g_{*}(\beta) > g_{*}(0)$ and the infimum in \eqref{eq:critical-value} is attained for sufficient small $\beta>0$ (see \cite{AtaLunNgu-24}). Therefore, the existence of minimizers of \eqref{functional:afP} are expected for $g > g_{*}(0)$.

The study of \eqref{energy:afP} offers a unique window into the behavior of quantum systems in both the collapse regime and the weak field regime, where the parameters $g$ and $\beta$ play crucial roles in shaping the system's properties. We here delve into the fascinating interplay between these regimes, seeking to unravel the intricate dynamics of $E^{\rm afP}_{\beta,g}$ and its minimizers as $g$ approaches its critical value $ g_{*}(0)$ while simultaneously taking $\beta$ to zero. As $g$ approaches $g_{*}(0)$, the system undergoes a transition from a stable state to a collapsed state, characterized by the onset of strong correlations and nontrivial phase transitions. Understanding the behavior of $E^{\rm afP}_{\beta,g}$ near this critical point is essential for elucidating the nature of these phase transitions and their implications for physical systems. Simultaneously, in the weak field regime where $\beta \to 0$, the external field exerts a diminishing influence on the system, allowing us to probe the intrinsic properties of the quantum particles and their interactions. Here, the competition between Pauli exclusion and external potentials becomes more pronounced, leading to intriguing phenomena such as the formation of bound states, localization effects, and emergence of topological features. Our main interest lies in bridging these two regimes and revealing the intricate relationship between collapse dynamics and weak field behavior in $E^{\rm afP}_{\beta,g}$. By analyzing the properties of the corresponding minimizers, we aim to gain deeper insights into the underlying mechanisms governing the system's behavior and its response to external perturbations.

Our first result concerns the existence and properties of minimizers of \eqref{energy:afP}.

\begin{theorem}[Existence of average-field-Pauli minimizers]\label{thm:existence}
Let $V \in L_{\rm loc}^{\infty}(\mathbb R^{2})$ and $V(x) \to +\infty$ as $|x| \to +\infty$. If $\beta > 0$ and $g < g_{*}(\beta)$ then there exists $u \in H^{1}(\mathbb R^{2})$ with $\int_{\mathbb R^{2}}|u|^{2}=1$ such that $\mathcal{E}^{\rm afP}_{\beta,g}[u]=E^{\rm afP}_{\beta,g}$.
\end{theorem}

\begin{remark}
    It is worth noting that minimizers for \eqref{energy:afP} exist not only for $g < g_{*}(0)$, as seen in the Gross--Pitaevskii theory of (rotating) Bose gases, but also for $g_{*}(\beta) > g \geq g_{*}(0)$ since $g_{*}(\beta) > g_{*}(0)$. This is the noteworthy of \Cref{thm:existence}. In fact, for large $\beta \geq 2$, the existence result was obtained for all $0<g<g_{*}(\beta) = 4\pi\beta$.
\end{remark}

Our next result concern the collapse phenomenon in the average-field-Pauli theory. In order to simplify and formula our next result, we here consider the trapping potential of the form
$$
V(x) = |x|^{s}
$$
for $s>0$. Furthermore, let's introduce the following notations
$$
Q_{0} := \|Q\|_{L^{2}}^{-1}Q,\quad \mathcal{Q}_{s} := \frac{s}{2}\int_{\mathbb R^{2}}|x|^{s}|Q_{0}|^{2} \quad \text{and} \quad \mathcal{A}_{0} = \|\mathbf{A}[|Q_{0}|^{2}] Q_{0}\|_{L^{2}}^{2}
$$
where $Q$ is the \emph{unique} (up to translation and dilation) real-valued solution of \eqref{eq:nls}.

\begin{theorem}[Condensation and collapse of average-field-Pauli minimizers]\label{thm:collapse-afP}
Let $g = g_{\beta}$ with $\beta \searrow 0$ and $\{u_{\beta}\}$ be a sequence of minimizers of $E^{\rm afP}_{\beta,g_{\beta}}$. Then there exists a sequence $\{\theta_{\beta}\}_{\beta} \subset [0, 2\pi)$ such that, for the whole sequence,
\begin{equation}\label{blowup:minimizer-afP}
\lim_{\beta \to 0}\ell_{\beta}u_{\beta}(\ell_{\beta}\cdot + x_{n}) e^{{\rm i}\theta_{n}} = Q_{0}
\end{equation}
strongly in $H^{1}(\mathbb R^{2})$. Moreover,
\begin{equation}\label{blowup:energy-afP}
E_{\beta,g_{\beta}}^{\rm afP} = \frac{s+2}{s}\ell_{\beta}^{s} \mathcal{Q}_{s}\left(1+o(1)_{\beta \to 0}\right).
\end{equation}
Here the blow-up length $\ell_{\beta}$ in \eqref{blowup:minimizer-afP} and \eqref{blowup:energy-afP} is determined in the following three cases:

\begin{enumerate}[label=(\roman*)]
\item\label{thm:collapse-subcritical} {\bf Sub-critical regime.} If $g_{\beta} \nearrow g_{*}(0)$ such that $\beta^{2} \ll g_{*}(0) - g_{\beta}$, then
\begin{equation}\label{length:subcritical}
\ell_{\beta} := \left(\frac{g_{*}(0) - g_{\beta}}{\mathcal{Q}_{s}g_{*}(0)}\right)^{\frac{1}{s+2}}.
\end{equation}

\item\label{thm:collapse-critical} {\bf Critical regime.} If $g_{\beta} = g_{*}(0)$, then
\begin{equation}\label{length:critical}
\ell_{\beta} := \left(\frac{\beta^{2}\mathcal{A}_{0}}{\mathcal{Q}_{s}}\right)^{\frac{1}{s+2}}
\end{equation}

\item\label{thm:collapse-supercritical} {\bf Super-criticial regime.} If $g_{\beta} = g_{*}(0) + \tau_{0}\beta^{2}$ where $0 < \tau_{0} < \tau_{*}$ and $\tau_{*}$ is a universal constant independent of $\beta$, then
\begin{equation}\label{length:supercritical}
\ell_{\beta} := \left(\frac{\beta^{2}}{\mathcal{Q}_{s}}\left(-\frac{\tau_{0}}{g_{*}(0)} + \mathcal{A}_{0}\right)\right)^{\frac{1}{s+2}} = \left(\frac{g_{\beta} - g_{*}(0)}{\mathcal{Q}_{s}}\left(-\frac{1}{g_{*}(0)} + \frac{\mathcal{A}_{0}}{\tau_{0}}\right)\right)^{\frac{1}{s+2}}.
\end{equation}
\end{enumerate}
 \end{theorem}

\begin{remark}
\begin{itemize}
\item The critical regime covered the case where $\beta^{2} \gg g_{*}(0) - g_{\beta}$. If $g_{*}(0) - g_{\beta} = \mathcal{O}(\beta^{2})$ (we assume without loss of generality that $g_{*}(0) - g_{\beta} = \widetilde{\tau}_{0}\beta^{2}$ for some $\widetilde{\tau}_{0} > 0$) then we still have \eqref{blowup:minimizer-afP} and \eqref{blowup:energy-afP} where $\ell_{\beta}$ is defined analogously as in \eqref{length:supercritical} with $\tau_{0}$ replaced by $\widetilde{\tau}_{0}$.

\item The assumption on $g_{\beta}$ in \Cref{thm:collapse-afP}-\ref{thm:collapse-supercritical} is to guarantee the existence of a minimizer. This is due to the fact that, for small $\beta$, the difference between the critical value $g_{*}(\beta)$ and the Gagliardo--Nirenberg constant $g_{*}(0)$ diminishes at a rate proportional to $\beta^{2}$ (see \cite{AtaLunNgu-24}). More precisely, there exist two constants $0 < C_{*} < C^{*} = g_{*}(0)\mathcal{A}_{0}$ (independent of $\beta$) such that
\begin{equation}\label{ineq:gn-weak-field}
g_{*}(0) + C_{*}\beta^{2} \leq g_{*}(\beta) \leq g_{*}(0) + C^{*}\beta^{2}.
\end{equation}
The universal constant $\tau_{*}$ in \Cref{thm:collapse-afP}-\ref{thm:collapse-supercritical} is then given by
$$
\tau_{*} := \liminf_{\beta \to 0}\frac{g_{*}(\beta) - g_{*}(0)}{\beta^{2}}.
$$
\end{itemize}
\end{remark}

\subsection{Many-body theory of interacting ``almost-bosonic'' anyon gases}

From the first principle of quantum mechanics, a 2D interacting anyon gas in a trapping potential $V$ is described by the system of $N$ particles (see \cite{LunRou-15})
\begin{align}\label{hamiltonian}
H^{\rm af}_{N^{\nu},\beta,g} ={} & \sum_{j=1}^N\left(\left(-{\rm i}\nabla_{j} + \frac{\beta}{N-1}\sum_{k\ne j}\nabla^{\perp} w_{0}(x_{j}-x_{k})\right)^{2} + V(x_j)\right) \nonumber \\
& - \frac{g}{N-1}\sum_{1 \leq j < k \leq N}W_{N^{\nu}}(x_{j}-x_{k})
\end{align}
acting on the Hilbert space $\mathcal{H}_{N} := L_{\rm sym}^{2}(\mathbb{R}^{2N})$. The two-body attraction is described by the function of distance between particles and is scales by the parameter $\nu>0$, given by
\begin{equation}\label{condition:potential-two-body}
0 \leq W_{N^{\nu}}(x) = N^{2\nu}W(N^{\nu}x) = W_{N^{\nu}}(-x) \in L^{1}\cap L^{\infty}(\mathbb R^{2}) \quad \text{with} \quad \int_{\mathbb R^{2}}W = 1.
\end{equation}
The quantum energy associated to \eqref{hamiltonian} is given by
\begin{equation}\label{energy:qm}
E_{N^{\nu},\beta,g}^{\rm afQM} := \inf\left\{\langle \Psi_{N} | H^{\rm af}_{N^{\nu},\beta,g} | \Psi_{N} \rangle : \Psi_{N} \in L^{2}(\mathbb R^{2N}), \int_{\mathbb R^{2N}} |\Psi_{N}|^{2} = 1 \right\}.
\end{equation}
Since \eqref{hamiltonian} takes action on the symmetric space, the usual mean-field approximation suggests to restrict many-body wave functions into the factorized ansatz
\begin{equation}\label{eq:bec}
\Psi_{N}(x_1,\ldots,x_{N}) \approx u^{\otimes N}(x_1,\ldots,x_{N}) := u(x_1)\ldots u(x_{N}).
\end{equation}
Taking the expectation of \eqref{hamiltonian} against \eqref{eq:bec}, with the normalization condition $\|u\|_{L^{2}}=1$, we obtain the ``average-field-Hartree'' energy functional
\begin{align}\label{functional:afH}
\mathcal{E}_{N^{\nu},\beta,g}^{\rm afH}[u] := &{} \frac{\langle u^{\otimes N}, H^{\rm af}_{N^{\nu},\beta,g} u^{\otimes N} \rangle}{N} \nonumber \\
={} & \int_{\mathbb R^{2}} \left[
\left|(-{\rm i}\nabla + \beta\mathbf{A}[|u|^{2}])u\right|^{2}
+ V|u|^{2} - \frac{g}{2} (|u|^{2} * W_{N^{\nu}})|u|^{2} \right].
\end{align}
The associated average-field-Hartree ground state energy, given by
\begin{equation}\label{energy:afH}
E_{N^{\nu},\beta,g}^{\rm afH} := \inf\left\{\mathcal{E}_{N^{\nu},\beta,g}^{\rm afH}[u] : u\in H^{1}(\mathbb R^{2}), \int_{\mathbb R^{2}} |u|^{2} = 1 \right\}
\end{equation}
is thus an upper bound to the many-body ground state energy. The average-field-Hartree theory plays an interpolation role between the many-body and one-body average-field-Pauli theory. In fact, the Hartree functional \eqref{functional:afH} is related to the average-field-Pauli functional \eqref{functional:afP} since the potential $W_{N^{\nu}}$ converges to the delta function $\delta_{0}$ in the mean-field limit regime $N \to \infty$.

It is worth noting that the Hamiltonian \eqref{hamiltonian} is too singular when acting on the product state \eqref{eq:bec}. To circumvent this issue, a regularization of the magnetic vector potential \eqref{potential:magnetic} was used in \cite{LunRou-15,Girardot-20}. It was given by (see \cite{LunRou-15,Girardot-20})
\begin{equation} \label{potential:magnetic-regularized}
\mathbf{A}_{R}[\varrho] := \nabla^{\perp} w_{R} * \varrho  \quad \text{with} \quad w_{R} := w_{0} * \frac{\mathbbm{1}(B(0,R))}{\pi^{2}R^{2}}.
\end{equation}
We then introduce the regularizations $\mathcal{E}^{\rm afP}_{R,\beta,g}$, $\mathcal{E}^{\rm afH}_{N^{\nu},R,\beta,g}$ and $E^{\rm afP}_{R,\beta,g}$,  $E^{\rm afH}_{N^{\nu},R,\beta,g}$ of the original energy functionals $\mathcal{E}^{\rm afP}_{\beta,g}$, $\mathcal{E}^{\rm afH}_{N^{\nu},\beta,g}$ and energies $E^{\rm afP}_{\beta,g}$, $E^{\rm afH}_{N^{\nu},\beta,g}$, where we replaced $\mathbf{A}$ (defined by \eqref{potential:magnetic}) in \eqref{functional:afP}, \eqref{functional:afH} and \eqref{energy:afP}, \eqref{energy:afH} respectively by $\mathbf{A}_{R}$ (defined by \eqref{potential:magnetic-regularized}). Furthermore, for the many-body problem, we consider the regularized Hamiltonian $H^{\rm af}_{N^{\nu},R,\beta,g}$ with the associated reguralized quantum energy $E^{\rm afQM}_{N^{\nu},R,\beta,g}$, where $w_{0}$ (defined by \eqref{potential:magnetic}) in \eqref{hamiltonian} was replaced by $w_{R}$ (defined by \eqref{potential:magnetic-regularized}).

In the literature, the derivation of the average-field-Pauli theory \eqref{functional:afP} as the limit of quantum mechanics was done for non-interacting anyon gases ($\beta \ne 0$ and $g = 0$) in \cite{LunRou-15} (see also \cite{Girardot-20}) and for interacting Bose gases ($\beta = 0$ and $g \ne 0$) in \cite{LewNamRou-16} (see also \cite{LewNamRou-17,NamRouSei-16,NamRou-20,Rougerie-EMS}. It is worth noting that, for Bose gases, the semi-classical approximation theory is usually called nonlinear Schr\"odinger (NLS) or Gross--Pitaevskii (GP). For interacting anyon gases, a weak result, in the high density regime, can be obtained by combining those well-known results. However, the most challenging problem is to obtain BECs in the dilute regime where $\nu > \frac{1}{2}$. It seems that the arguments using the second moment estimate in \cite{NamRouSei-16,LewNamRou-17,NamRou-20} cannot be used since the Hamiltonian \eqref{hamiltonian} invoked three-body interactions. Such an argument is only available in the low one-dimensional setting (see \cite{NguRic-24} for a discussion).

In order to verify the validity of the effective average-field-Pauli theory \eqref{functional:afP} from the many-body quantum mechanic, it is necessary to introduce the $k$-particles reduced density matrices. It is defined, for any many-body wave function $\Psi_{N} \in L^{2}(\mathbb R^{2N})$, by the partial trace
$$
\gamma_{\Psi_{N}}^{(k)} := \tr_{k+1\to N} | \Psi_{N} \rangle \langle \Psi_{N} |,\quad \forall k \in \mathbb N .
$$
Equivalently, $\gamma_{\Psi_{N}}^{(k)}$ is the trace class operator on $L^{2}(\mathbb R^{2k})$ with kernel
$$
\gamma_{\Psi_{N}}^{(k)}(x_1,\ldots,x_k;y_1,\ldots,y_k) := \int_{\mathbb R^{2(N-k)}}\overline{\Psi_{N}(x_1,\ldots,x_k;Z)}{\Psi_{N}(y_1,\ldots,y_k;Z)} {\rm d} Z .
$$
The Bose--Einstein condensation \eqref{eq:bec} is then characterized properly by
$$
\lim_{N \to \infty} \tr \left|\gamma_{\Psi_{N}}^{(k)} - |u^{\otimes k}\rangle \langle u^{\otimes k}|\right| = 0,\quad \forall k \in \mathbb N.
$$
For (non-)interacting anyon gases, one of the main features of the reduced density matrices is to rewrite the quantum energy, i.e., the expectation of the Hamiltonian \eqref{hamiltonian} in terms of three-particles reduced density matrices (see \cite{LunRou-15,LewNamRou-16}). We are now in the position to formulate our last main result. We have the following.

\begin{theorem}[Condensation and collapse of many-body ground states]\label{thm:many-body}
We consider $N$ extended interacting anyons $E^{\rm afQM}_{N^{\nu},R,\beta,g}$ of radius $R = R_{N} \sim N^{-\eta}$ in an external potential $V(x) = |x|^{s}$ with $s>0$. We assume that the two-body interaction function $W_{N^{\nu}}$ satisfy \eqref{condition:potential-two-body} and assume the relation
\begin{equation}\label{cv:speed-many-body}
\nu < \frac{s}{6s+6} \quad \text{and} \quad 0 < \eta < \frac{s}{4s+4}.
\end{equation}
\begin{enumerate}[label=(\roman*)]
\item Let $\beta,g > 0$ be fixed with $g < g_{*}(\beta)$. Assume further that $\nu<\eta$ if $g \geq g_{*}(0)$. Then, in the limit $N \to \infty$, we have the convergence of the ground state energy
\begin{align}\label{cv:qm-afP-energy}
\lim_{N\to \infty} E^{\rm afQM}_{N^{\nu},R_{N},\beta,g} = E^{\rm afP}_{\beta,g} \geq 0 .
\end{align}
Moreover, if $\{\Psi_{N}\}_{N}$ is a sequence of ground states of  $E^{\rm afQM}_{N^{\nu},R_{N},\beta,g}$, then there exists a Borel probability measure $\mu$ supported on the set $\mathcal{M}^{\rm afP}$ of minimizers of $\mathcal{E}^{\rm afP}_{\beta,g}$, i.e.,
$$
\mathcal{M}^{\rm afP} := \left\{u \in L^{2}(\mathbb R^{2}): \mathcal{E}^{\rm afP}_{\beta,g}[u] = E^{\rm afP}_{\beta,g}, \int_{\mathbb R^{2}}|u|^{2} = 1\right\}
$$
such that, modulo restricting to a subsequence, we have
\begin{equation}\label{cv:qm-afP-ground-state}
\lim_{N \to \infty}\tr \left| \gamma_{\Psi_{N}}^{(k)} - \int_{\mathcal{M}^{\rm afP}} |u^{\otimes k} \rangle \langle u^{\otimes k}| {\rm d}\mu(u) \right| =0,\quad \forall k \in \mathbb N.
\end{equation}
\item Let $Q_{0}$, $\mathcal{Q}_{s}$, $\mathcal{A}_{0}$, $\{\beta_{N}\}_{N}$ with $\beta_{N} \searrow 0$, $\{g_{N}\}_{N}$ with $g_{N} \to g_{*}(0)$, and $\{\ell_{N}\}_{N}$ be as in \Cref{thm:collapse-afP}. Assume that $xW \in L^1(\mathbb R^{2})$ and $\ell_{N} \sim N^{-\sigma}$ with
\begin{equation}\label{blowup:speed-many-body}
	0<\sigma<\min\left\{\frac{\nu}{s+3},\frac{2\eta}{s+4},\frac{s-\nu(6s+6)}{s(5s+4)},\frac{s-\eta(4s+4)}{s(5s+4)}\right\}.
\end{equation}
Assume further that  $\nu<\eta$ in the critical case and additionally $\sigma < \frac{2(\eta-\nu)}{s+2}$ in the super-critical case. Then we have the asymptotic formula of the ground state energy
\begin{equation}\label{blowup:qm-energy}
	E_{N,R_{N},\beta_{N},g_{N}}^{\rm afQM} = E_{\beta_{N},g_{N}}^{\rm afP} \left(1+o(1)_{N \to \infty}\right) = \frac{s+2}{s} \mathcal{Q}_{s} \ell_{N}^{s} \left(1+o(1)_{N \to \infty}\right) .
\end{equation}
Moreover, if $\{\Psi_{N}\}_{N}$ is a sequence of ground states of $E^{\rm afQM}_{N^{\nu},R_{N},\beta_{N},g_{N}}$ and if in addition 
$$
0 < \sigma < \min\left\{\frac{s-\nu(6s+6)}{5s^{2}+12s+8},\frac{s-\eta(4s+4)}{5s^{2}+12s+8}\right\}
$$
then for the whole sequence $\Phi_{N} = \ell_{N}^{2N} \Psi_{N}(\ell_{N}\cdot)$, we have
\begin{equation}\label{blowup:qm-ground-state}
	\lim_{N \to \infty} \tr \left| \gamma_{\Phi_{N}}^{(k)} - | Q_{0}^{\otimes k} \rangle \langle Q_{0}^{\otimes k} | \right| = 0,\quad \forall k \in \mathbb N.
\end{equation}
\end{enumerate}
\end{theorem}

\begin{remark}
If $\mathcal{M}^{\rm afP} = \{u_{0}\}$, up to a phase, then for the whole sequence,
$$
\lim_{N \to \infty}\tr \left| \gamma_{\Psi_{N}}^{(k)} - |u_{0}^{\otimes k} \rangle \langle u_{0}^{\otimes k}| \right| =0,\quad \forall k \in \mathbb N.
$$
\end{remark}

In this article, the method of proof of \Cref{thm:many-body} is the (quantitative) quantum de Finetti theorem as used in \cite{LewNamRou-16,LunRou-15}. Such a method was also used in the study of 2D Bose gases with non-mangetic field \cite{LewNamRou-17} (see also \cite{NamRou-20} for an improvement) in the dilute regime. However, arguments in the mentioned articles cannot be used for the problem of three-body interactions in high-dimensional systems such as anyon gases (see \cite{NguRic-24} for discussions in the low one-dimensional case). Fortunately, one can still use some arguments in \cite{LewNamRou-16,LunRou-15} to obtain a weak result in the high density regime.

\hspace{1cm}

\noindent{\bf Organization of the paper.} In \Cref{sec:afP}, we prove the existence of average-field-Pauli minimizers and establish its behavior in the collapse regime and in the weak-field regime at the same time. In \Cref{sec:many-body}, we first study the collapse phenomenon in the interpolate average-field-Hartree theory. This plays a crucial role in the study of many-body theory in the super-critical regime. Finally, in \Cref{app:gn}, we make a relation between the magnetic and non-magnetic Gagliardo--Nirenberg optimal constant and optimizer(s) in the weak-field limit. This gives a hope to expect the uniqueness of magnetic Gagliardo--Nirenberg optimizers for small enough self-generated magnetic field.

\hspace{1cm}

\noindent{\bf Acknowledgments.} The author is supported by the Knut and Alice Wallenberg Foundation.

\section{Average-field-Pauli minimizers}\label{sec:afP}

\subsection{Existence of average-field-Pauli minimizers}

In this subsection, we prove the existence of average-field-Pauli minimizers as shown in \Cref{thm:existence}. We rewrite the energy functional \eqref{functional:afP} as follows
\begin{equation}\label{functional:afP-crossterms}
\mathcal{E}^{\rm afP}_{\beta,g}[u] = \int_{\mathbb R^{2}} \left[
|\nabla u|^{2} + 2\beta \mathbf{A}[|u|^{2}] \cdot \mathbf{J}[u] 
+ \beta^{2} \mathbf{A}[|u|^{2}]^{2} |u|^{2} + V |u|^{2} - \frac{g}{2} |u|^{4}
\right]
\end{equation}
where we defined the current of $u$ by
\begin{equation}\label{potential:current}
\mathbf{J}[u] := \frac{{\rm i}}{2}(u\nabla \overline{u} - \overline{u}\nabla u) = \Im(\overline{u}\nabla u).
\end{equation}
We now collect some preliminary estimates that are needed to set up the variational calculus for \eqref{functional:afP-crossterms}. We note that, for $u \in H^{1}(\mathbb R^{2})$, by Sobolev embedding,
$u \in L^{p}(\mathbb R^{2})$ for any $2 \le p < \infty$. Furthermore, $\mathbf{A}[|u|^{2}]u \in L^{2}(\mathbb R^{2})$ (see \cite{LunRou-15,CorLunRou-17}) with
\begin{equation}\label{ineq:hardy}
\|\mathbf{A}[|u|^{2}]u\|_{L^{2}}
\le C\|u\|_{L^{2}}^{2} \|\nabla |u|\|_{L^{2}},\quad \forall u \in H^{1}(\mathbb R^{2}).
\end{equation}
Note that \eqref{ineq:hardy} still holds true if the magnetic field $\mathbf{A}$ (defined by \eqref{potential:magnetic}) is replaced by its regularization $\mathbf{A}_{R}$ (defined by \eqref{potential:magnetic-regularized}). Moreover, we have the following 2D Gagliardo--Nirenberg inequality
\begin{equation}\label{ineq:gn}
\frac{g_{*}(0)}{2}\int_{\mathbb R^{2}}|u|^{4} \leq \int_{\mathbb R^{2}}|\nabla |u||^{2} \int_{\mathbb R^{2}}|u|^{2} \leq \int_{\mathbb R^{2}}|\nabla u|^{2} \int_{\mathbb R^{2}}|u|^{2},\quad \forall u\in H^{1}(\mathbb R^{2})
\end{equation}
and the diamagnetic inequality (see e.g., \cite[Theorem 7.21]{LieLos-01})
\begin{equation}\label{ineq:diamagnetic}
\int_{\mathbb R^{2}}|\nabla |u||^{2} \leq \min\left\{\int_{\mathbb R^{2}}|\nabla u|^{2}, \int_{\mathbb R^{2}}\left|(-{\rm i}\nabla + \beta \mathbf{A}[|u|^{2}])u\right|^{2}\right\},\quad \forall \beta \in \mathbb{R}.
\end{equation}
Note that \eqref{ineq:gn} is a consequence of the Sobolev embedding and is the special case of \eqref{ineq:gn-magnetic} where $\beta=0$. Furthermore, \eqref{ineq:diamagnetic} still holds true for a general magnetic field $\mathbf{A}$ including \eqref{potential:magnetic} and \eqref{potential:magnetic-regularized}.

We have all tools to derive the existence of minimizers of \eqref{energy:afP}.

\begin{proposition}
Let $0 \leq V \in L_{\rm loc}^{\infty}(\mathbb R^{2})$ and $V(x) \to +\infty$ as $|x| \to +\infty$. Then, for every $\beta \geq 0$ and $0 < g < g_{*}(\beta)$, there exists $u \in H^{1}(\mathbb R^{2})$ with $\int_{\mathbb R^{2}}|u|^{2}=1$ such that $\mathcal{E}^{\rm afP}_{\beta,g}[u]=E^{\rm afP}_{\beta,g}$.
\end{proposition}

\begin{proof}
The existence of a minimizer follows from the standard method in the calculus of variation. While the case $\beta > 0$ and $g = 0$ was done in \cite{LunRou-15}, the case $\beta = 0$ and $0<g<g_{*}(0)$ was proved in \cite{GuoSei-14}. When $\beta > 0$ and $0 < g < g_{*}(\beta)$, it is enough to verify the uniform boundedness in $H^{1}\cap L_{V}^{2}(\mathbb R^{2})$ (where we recall the definition of the space $L_{V}^{2}$ in \eqref{space:external}) of the minimizing sequence $\{u_{n}\}$ given by
$$
\lim_{n \to \infty} \mathcal{E}^{\rm afP}_{\beta,g}[u_{n}] = E^{\rm afP}_{\beta,g} \quad \text{with} \quad \int_{\mathbb R^{2}}|u_{n}|^{2}=1.
$$
We use \eqref{ineq:gn} and \eqref{ineq:diamagnetic} to obtain
$$
\mathcal{E}^{\rm afP}_{\beta,g}[u_{n}] \geq \left(1 - \frac{g}{g_{*}(\beta)}\right) \mathcal{E}^{\rm afP}_{\beta,0}[u_{n}] + \int_{\mathbb R^{2}} V|u|^{2}.
$$
The above together with \eqref{ineq:hardy} and \eqref{ineq:diamagnetic} yields that $\{u_{n}\}$ is uniformly bounded in $H^{1}\cap L_{V}^{2}(\mathbb R^{2})$. In fact, the case where $g = g_{*}(0)$ was included since $g_{*}(\beta) > g_{*}(0)$ for every $\beta > 0$ (see \cite{AtaLunNgu-24}). However, the proof of it in \cite{AtaLunNgu-24} is not a direct proof. In the following, we give an alternative proof for this critical case where $g=g_{*}(0)$. 

Using \eqref{ineq:gn} and \eqref{ineq:diamagnetic}, we have that $\{u_{n}\}$ is uniformly bounded in $L_{V}^{2}(\mathbb R^{2})$. It remains to prove that
$$
\varepsilon_{n}^{-2} := \int_{\mathbb R^{2}} \left| (-{\rm i}\nabla + \beta\mathbf{A}[|u_{n}|^{2}]) u_{n}\right|^{2}
$$
is bounded uniformly. We assume on the contrary that there exists a subsequence of $\{u_{n}\}$ (still denoted by $\{u_{n}\}$) such that $\varepsilon_{n} \to 0$. Denote $\widetilde{u_{n}} = \varepsilon_{n} u_{n}(\varepsilon_{n} \cdot)$. Then we have 
\begin{equation}\label{cv:minimizer-critical-boundedness-2}
\int_{\mathbb R^{2}}|\widetilde{u_{n}}|^{2} = 1 = \int_{\mathbb R^{2}} \left| (-{\rm i}\nabla + \beta\mathbf{A}[|\widetilde{u_{n}}|^{2}])\widetilde{u_{n}}\right|^{2}
\end{equation}
and, by dropping the nonnegative trapping potential,
$$
\mathcal{E}^{\rm afP}_{\beta,g_{*}(0)}[u_{n}] \geq \varepsilon_{n}^{-2}\int_{\mathbb R^{2}} \left[\left|(-{\rm i}\nabla + \beta\mathbf{A}[|\widetilde{u_{n}}|^{2}])\widetilde{u_{n}}\right|^{2} - \frac{g_{*}(0)}{2}|\widetilde{u_{n}}|^{4}\right].
$$
Multiplying both side by $\varepsilon_{n}^{2}$, dropping the nonnegative trapping term and using \eqref{ineq:diamagnetic} we have
\begin{equation}\label{cv:minimizer-critical-boundedness-1}
\lim_{n \to \infty}\int_{\mathbb R^{2}} \left[\left|(-{\rm i}\nabla + \beta\mathbf{A}[|\widetilde{u_{n}}|^{2}])\widetilde{u_{n}}\right|^{2} - \frac{g_{*}(0)}{2}|\widetilde{u_{n}}|^{4}\right] = 0.
\end{equation}
Using \eqref{ineq:hardy}, \eqref{ineq:diamagnetic}, we deduce from \eqref{cv:minimizer-critical-boundedness-2} that $\{\widetilde{u_{n}}\}$ is bounded uniformly in $H^{1}(\mathbb R^{2})$. Up to translation and extracting a subsequence, we have $\widetilde{u_{n}} \to \widetilde{u_{0}}$ weakly in $H^{1}(\mathbb R^{2})$ and almost everywhere in $\mathbb R^{2}$. We prove that this is actually the strong convergence in $H^{1}(\mathbb R^{2})$. We first note that $\widetilde{u_{0}} \not\equiv 0$. Otherwise, we must have that $u_{n} \to 0$ strongly in $L^{p}(\mathbb R^{2})$ for all $2<p<\infty$ (see \cite{Lions-84b}). This yields a contradiction, by \eqref{cv:minimizer-critical-boundedness-1} and \eqref{cv:minimizer-critical-boundedness-2}. Next, we eliminate the case where $0 < \|\widetilde{u_{0}}\|_{L^{2}}^{2} < 1$. We assume on the contrary that this is the case. It is worth noting that $\{|\widetilde{u_{n}}|\}$ is also bounded uniformly in $H^{1}(\mathbb R^{2})$, by the diamagnetic inequality \eqref{ineq:diamagnetic}. Then $|\widetilde{u_{n}}| \to |\widetilde{u_{0}}|$ weakly in $H^{1}(\mathbb R^{2})$ and almost everywhere in $\mathbb R^{2}$. We deduce from \eqref{cv:minimizer-critical-boundedness-1} and \eqref{ineq:diamagnetic} and Brezis--Lieb lemma \cite{BreLie-83} that
\begin{align}\label{cv:minimizer-critical-dichotomy-1}
0 \geq{} & \int_{\mathbb R^{2}} \left[\left|\nabla|\widetilde{u_{n}}|\right|^{2} - \frac{g_{*}(0)}{2}|\widetilde{u_{n}}|^{4}\right] \nonumber \\
={} & \int_{\mathbb R^{2}} \left[\left|\nabla|\widetilde{u_{0}}|\right|^{2} - \frac{g_{*}(0)}{2}|\widetilde{u_{0}}|^{4}\right] + \lim_{n \to \infty}\int_{\mathbb R^{2}} \left[\left|\nabla(|\widetilde{u_{n}}|-|\widetilde{u_{0}}|)|\right|^{2} - \frac{g_{*}(0)}{2}\left||\widetilde{u_{n}}|-|\widetilde{u_{0}}|\right|^{4}\right] \nonumber \\
={} & I_{1} + I_{2}.
\end{align}
On one hand, it follows from \eqref{ineq:gn} and \eqref{ineq:diamagnetic} that
\begin{equation}\label{cv:minimizer-critical-dichotomy-2}
I_{1} \geq \left(\frac{1}{\int_{\mathbb R^{2}}|\widetilde{u_{0}}|^{2}}-1\right) \frac{g_{*}(0)}{2} \int_{\mathbb R^{2}}|\widetilde{u_{0}}|^{4} > 0,
\end{equation}
On the other hand, by again \eqref{ineq:gn} and \eqref{ineq:diamagnetic},
\begin{equation}\label{cv:minimizer-critical-dichotomy-3}
I_{2} \geq \lim_{n \to \infty}\left(\frac{1}{\int_{\mathbb R^{2}}|\widetilde{u_{n}}-\widetilde{u_{0}}|^{2}}-1\right) \frac{g_{*}(0)}{2} \int_{\mathbb R^{2}}|\widetilde{u_{n}}-\widetilde{u_{0}}|^{4} \geq 0.
\end{equation}
Combining \eqref{cv:minimizer-critical-dichotomy-1}, \eqref{cv:minimizer-critical-dichotomy-2} and \eqref{cv:minimizer-critical-dichotomy-3} we obtain a contraction. Therefore, we must have that $\|\widetilde{u_{0}}\|_{L^{2}}^{2} = 1$. This yields that, by Brezis--Lieb lemma \cite{BreLie-83}, up to translation and extracting a subsequence, $\widetilde{u_{n}} \to \widetilde{u_{0}}$ strongly in $L^{2}(\mathbb R^{2})$. In fact, this strong convergence holds in $L^{p}(\mathbb R^{2})$ for any $2 \leq p < \infty$, by interpolation. We then deduce from \eqref{cv:minimizer-critical-boundedness-1} that
\begin{align}\label{cv:minimizer-critical-dichotomy-4}
\frac{g_{*}(0)}{2}\int_{\mathbb R^{2}}|\widetilde{u_{0}}|^{4} = \frac{g_{*}(0)}{2}\lim_{n \to \infty}\int_{\mathbb R^{2}}|\widetilde{u_{n}}|^{4} & = \lim_{n \to \infty} \int_{\mathbb R^{2}} \left| (-{\rm i}\nabla + \beta\mathbf{A}[|\widetilde{u_{n}}|^{2}])\widetilde{u_{n}}\right|^{2} \nonumber \\
& \geq \int_{\mathbb R^{2}} \left| (-{\rm i}\nabla + \beta\mathbf{A}[|\widetilde{u_{0}}|^{2}])\widetilde{u_{0}}\right|^{2}.
\end{align}
Here we have used Fatou's lemma (see e.g., \cite[Proposition 3.8]{LunRou-15}) in the last inequality. Then \eqref{ineq:gn} yields that we must have equality in the above. This, however, is not possible. Indeed, assume on the contrary that the equality in the above holds true. Then, by \eqref{ineq:gn} and \eqref{ineq:diamagnetic}, $|\widetilde{u_{0}}| = Q_{0}$ where $Q_{0}$ is the normalized unique solution (up to translation and dilation) of the equation \eqref{eq:nls}. On the other hand, we can write $\widetilde{u_{0}}=|\widetilde{u_{0}}| e^{{\rm i} \varphi}$ where $\varphi$ is a real-valued function. Then we have
\begin{align*}
\left|(-{\rm i}\nabla + \beta \mathbf{A}[|\widetilde{u_{0}}|^{2}])\widetilde{u_{0}}\right|^{2} 
& = \left|-{\rm i}\nabla|\widetilde{u_{0}}| + |\widetilde{u_{0}}|(\nabla \varphi+\beta \mathbf{A}[\widetilde{u_{0}}|^{2}])\right|^{2} \\
& = \left|\nabla |\widetilde{u_{0}}|\right|^{2} + |\widetilde{u_{0}}|^{2} \left|\nabla \varphi + \beta \mathbf{A}[|\widetilde{u_{0}}|^{2}]\right|^{2} \\
& = \left|\nabla Q_{0}\right|^{2} + Q_{0}^{2} \left|\nabla \varphi + \beta \mathbf{A}[Q_{0}^{2}]\right|^{2}
\end{align*}
Therefore, we deduce from the equality in \eqref{cv:minimizer-critical-dichotomy-4} that we must have
$$
\int_{\mathbb R^{2}} Q_{0}^{2} \left|\nabla \varphi + \beta \mathbf{A}[Q_{0}^{2}]\right|^{2} = 0,
$$
which yields that
$$
\nabla \varphi + \beta \mathbf{A}[Q_{0}^{2}] = 0
$$
almost everywhere in $\mathbb R^{2}$. By taking skew gradient both side in the above and using \eqref{eq:self-field}, we obtain that
$$
2\pi\beta Q_{0}^{2} = 0.
$$
This is impossible whence $\beta > 0$. Therefore, the equality in \eqref{cv:minimizer-critical-dichotomy-4} cannot occur and we must have that $\{u_{n}\}$ is uniformly bounded in $H^{1}(\mathbb R^{2})$. At this stage, the proof of existence of average-field-Pauli minimizers in the critical case where $g = g_{*}(0)$ follows from the compactness of $\{u_{n}\}$ in $H^{1}\cap L_{V}^{2}(\mathbb R^{2})$.
\end{proof}

\begin{remark}
By using \eqref{ineq:diamagnetic} with $\mathbf{A}$ replaced by $\mathbf{A}_{R}$, one can still prove the existence of minimizers for the regularized minimization problem. In other words, for every $\beta, R > 0$ and $g < g_{*}(0)$, there exists $u \in H^{1}(\mathbb R^{2})$ with $\int_{\mathbb R^{2}}|u|^{2}=1$ such that $\mathcal{E}^{\rm afP}_{R,\beta,g}[u]=E^{\rm afP}_{R,\beta,g}$. It leaves an open question whether there still exists a (regularized) minimizer in the region $g_{*}(0) \leq g < g_{*}(\beta)$. The (regularized) magnetic Gagliardo--Nirenberg inequality \eqref{ineq:gn-magnetic}, where $\mathbf{A}$ replaced by $\mathbf{A}_{R}$, might not be validated anymore.
\end{remark}

\subsection{Collapse of average-field-Pauli minimizers in the weak field regime}

In this section, we establish the collapse phenomenon in the average-field-Pauli theory. We prove the asymptotic formula of the average-field-Pauli energy and it minimizers as shown in \Cref{thm:collapse-afP}. We recall that, for simplicity, we take in particular $V(x) = |x|^{s}$ for $s>0$.

\subsubsection{The sub-critical case}\label{sec:subcritical-afP}

In this subsection, we establish the behavior of $E^{\rm afP}_{\beta,g}$ given by \eqref{functional:afP} as well as its minimizers when $\beta \searrow 0$ slower than $g = g_{\beta} \nearrow g_{*}(0)$. Essentially, this is similar to the case where $\beta = 0$ which was done in \cite{GuoSei-14}. For reader's convenience, we give a detail proof in the following.

By the variational principle, we have
\begin{equation}\label{cv:energy-afP-subcritical-upper-bound}
E_{\beta,g_{\beta}}^{\rm afP} \leq \mathcal{E}^{\rm afP}_{\beta,g_{\beta}}[\ell^{-1} Q_{0}(\ell^{-1}\cdot)] = \ell^{-2}\left(1-\frac{g_{\beta}}{g_{*}(0)}\right) + \ell^{s}\frac{2}{s}\mathcal{Q}_{s} + \ell^{-2}\beta^{2}\mathcal{A}_{0}
\end{equation}
for any $\ell>0$. Since $\beta^{2} \ll g_{*}(0) - g_{\beta}$, the last term on the right hand side of \eqref{cv:energy-afP-subcritical-upper-bound} is of small order to the leading order term (the first term on the right hand side of \eqref{cv:energy-afP-subcritical-upper-bound}). By optimizing over all $\ell>0$ the first two terms on the right hand side of \eqref{cv:energy-afP-subcritical-upper-bound}, we obtain the energy upper bound in \eqref{blowup:energy-afP}. We prove the energy lower bound by justifying the convergence of minimizers in \eqref{blowup:minimizer-afP}.

Let's denote $w_{\beta}=\ell_{\beta}u_{\beta}(\ell_{\beta}\cdot)$ where $\ell_{\beta}$ is given by \eqref{length:subcritical}. Then $\|w_{\beta}\|_{L^{2}} = \|u_{\beta}\|_{L^{2}} = 1$ and
\begin{align}
\mathcal{E}^{\rm afP}_{\beta,g_{\beta}}[u_{\beta}] & = 
\ell_{\beta}^{-2}\int_{\mathbb R^{2}} \left[ \left| (-{\rm i}\nabla + \beta\mathbf{A}[|w_{\beta}|^{2}]) w_{\beta}\right|^{2} - g_{\beta} |w_{\beta}|^{4} \right] + \ell_{\beta}^{s}\int_{\mathbb R^{2}} V|w_{\beta}|^{2} \label{cv:minimizer-afP-nls-subcritical-1} \\
& \geq \ell_{\beta}^{-2}\left(1-\frac{g_{\beta}}{g_{*}(0)}\right)\int_{\mathbb R^{2}} \left|(-{\rm i}\nabla + \beta\mathbf{A}[|w_{\beta}|^{2}])w_{\beta}\right|^{2} + \ell_{\beta}^{s}\int_{\mathbb R^{2}} V|w_{\beta}|^{2} \label{cv:minimizer-afP-nls-subcritical-2}
\end{align}
where we have used \eqref{ineq:gn} and \eqref{ineq:diamagnetic}. It follows from \eqref{cv:minimizer-afP-nls-subcritical-2} and the energy upper bound in \eqref{blowup:energy-afP} that
\begin{equation}\label{cv:minimizer-afP-subcritical-boundedness}
\frac{s+2}{s}\left(1+o(1)_{\beta \to 0}\right) \geq \int_{\mathbb R^{2}} \left| (-{\rm i}\nabla + \beta\mathbf{A}[|w_{\beta}|^{2}]) w_{\beta}\right|^{2} + \mathcal{Q}_{s}^{-1}\int_{\mathbb R^{2}} V|w_{\beta}|^{2}.
\end{equation}
We deduce from \eqref{cv:minimizer-afP-subcritical-boundedness}, \eqref{ineq:hardy} and \eqref{ineq:diamagnetic} that $\{w_{\beta}\}$ is uniformly bounded in $H^{1} \cap L_{V}^{2} (\mathbb R^{2})$. After extracting a subsequence, we have $w_{\beta} \to w_{0}$ almost everywhere in $\mathbb R^{2}$, weakly in $H^{1}(\mathbb R^{2})$ and strongly in $L^{p}(\mathbb R^{2})$ for any $2 \leq p < \infty$. Furthermore,
\begin{align}\label{cv:minimizer-afP-weak-lower-semicontinuous}
\|\nabla w_{0}\|_{L^{2}} = \sup _{\|v\|_{L^{2}}=1}|\langle\nabla w_{0}, v\rangle| & = \sup_{\|v\|_{L^{2}}=1} \lim _{\beta \to 0}\left|\left\langle (-{\rm i}\nabla + \beta \mathbf{A}[|w_{\beta}|^{2}]) w_{\beta}, v\right\rangle\right| \nonumber \\
& \leq \liminf_{\beta \to 0}\left\|(-{\rm i}\nabla +\beta \mathbf{A}[|w_{\beta}|^{2}]) w_{\beta}\right\|_{L^{2}}
\end{align}
Multiplying both side of \eqref{cv:minimizer-afP-nls-subcritical-1} by $\ell_{\beta}^{2}$, dropping the nonnegative external term, using again the energy upper bound in \eqref{blowup:energy-afP} and taking the limit $\beta \to 0$, we have
$$
0 \geq \lim_{\beta \to 0}\int_{\mathbb R^{2}} \left[ \left| (-{\rm i}\nabla + \beta\mathbf{A}[|w_{\beta}|^{2}]) w_{\beta}\right|^{2} - g_{\beta} |w_{\beta}|^{4} \right] \geq \int_{\mathbb R^{2}} \left[ \left| \nabla v_{0}\right|^{2} - g_{*}(0) |v_{0}|^{4} \right] \geq 0.
$$
Here we have used \eqref{cv:minimizer-afP-weak-lower-semicontinuous} in the next-to-last inequality and \eqref{ineq:diamagnetic} in the last inequality. Therefore, we must have equality in the above and it yields that $v_{0}$ is an optimizer of \eqref{ineq:diamagnetic}. Hence
$$
w_{0}(x) = t_{0}Q_{0}(t_{0}x + x_{0})
$$
for some $t_{0}>0$, $x_{0}\in\mathbb R^{2}$ and $Q_{0}$ being the (unique) real-valued solution of \eqref{eq:nls}. We prove that $v_{0} \equiv Q_{0}$ by verifying that $t_{0}=1$ and $x_{0} \equiv 0$. Indeed, taking the limit $\beta \to 0$ in \eqref{cv:minimizer-afP-subcritical-boundedness} and using \eqref{cv:minimizer-afP-weak-lower-semicontinuous}, we get
\begin{align}\label{cv:minimizer-afP-nls-subcritical-3}
\frac{s+2}{s} & \geq \int_{\mathbb R^{2}} |\nabla w_{0}|^{2} + \mathcal{Q}_{s}^{-1}\int_{\mathbb R^{2}} |x|^{s}|w_{0}|^{2} \nonumber \\
& = t_{0}^{2}\int_{\mathbb R^{2}} |\nabla Q_{0}|^{2} + t_{0}^{-s}\mathcal{Q}_{s}^{-1}\int_{\mathbb R^{2}} |x|^{s}|Q_{0}(x-x_{0})|^{2} \nonumber \\
& \geq t_{0}^{2} + \frac{2}{s}t_{0}^{-s}.
\end{align}
Here we have used the rearrangement inequality as $Q_{0}$ is symmetric decreasing and $|x|^{s}$ is strictly symmetric increasing. It is elementary to check that
$$
\inf_{r>0} \left(r^{2} + \frac{2}{s}r^{-s}\right) = \frac{s+2}{s}
$$
and the equality is attained at $r_{0}=1$. Therefore, we must have equality in \eqref{cv:minimizer-afP-nls-subcritical-3} and it yields that $t_{0} = r_{0} = 1$ and $x_{0}=0$. This conclude the energy lower bound in \eqref{blowup:energy-afP} as well as the convergence of blow-up sequence in \eqref{blowup:minimizer-afP}.

\subsubsection{The critical case}\label{sec:critical-afP}

We now establish the behavior of $E^{\rm afP}_{\beta,g}$ given by \eqref{functional:afP} as well as its minimizers at the critical value $g = g_{*}(0)$. In which case, the existence of minimizers can be obtained easily since the critical interaction strength $g_{*}(\beta)$ is strictly larger than $g_{*}(0)$. However, the analysis on the collapse of minimizers seems complicated due to the lack of its compactness in the collapse regime. The strategy of proof follows from \cite{DinNguRou-23}.

By the variational principle, we have
\begin{equation}\label{cv:energy-afP-upper-bound-critical}
E_{\beta,g_{*}(0)}^{\rm afP} \leq \mathcal{E}^{\rm afP}_{\beta,g_{*}(0)}[\ell^{-1} Q_{0}(\ell^{-1}\cdot)] = \ell^{s}\frac{2}{s}\mathcal{Q}_{s} + \ell^{-2}\beta^{2}\mathcal{A}_{0}
\end{equation}
for any $\ell>0$. By optimizing \eqref{cv:energy-afP-upper-bound-critical} over all $\ell>0$, we obtain the energy upper bound in \eqref{blowup:energy-afP}. We prove the energy lower bound by justifying the convergence of minimizers in \eqref{blowup:minimizer-afP}. The proof is divided into several steps as in \cite{DinNguRou-23} with a few simplified arguments and is processed as follows.

\vspace{10pt}

\noindent {\bf Step 1 (Blow-up property).}  The sequence $\{u_{\beta}\}_{\beta}$ blows up in $H^{1}(\mathbb R^{2})$ in the sense that
\begin{equation}\label{blowup:property-critical}
\lim_{\beta \to 0}\|\nabla u_{\beta}\|_{L^{2}} = +\infty.
\end{equation}
Indeed, assume for contradiction that $\{u_{\beta}\}_{\beta}$ is a bounded sequence in $H^{1}(\mathbb R^{2})$. Using \eqref{ineq:gn} and \eqref{ineq:diamagnetic}, we have that $\{u_{\beta}\}_{\beta}$ is also uniformly bounded in $L_{V}^{2}$. Therefore, up to extraction a subsequence, $u_{\beta} \to u_{0}$ strongly in $L^{p}(\mathbb R^{2})$ for any $2 \leq p < \infty$. It then follows from the upper bound in \eqref{blowup:energy-afP} that
$$
0 \geq \lim_{\beta \to 0} \mathcal{E}^{\rm afP}_{\beta,g_{*}(0)}[u_{\beta}] \geq \mathcal{E}^{\rm afP}_{0,g_{*}(0)}[u_{0}] \geq E_{0,g_{*}(0)}^{\rm afP} = 0.
$$
This yields that $u_{0}$ is a minimizer for $E_{0,g_{*}(0)}^{\rm afP} = 0$ which, however, is impossible. Therefore, we must have \eqref{blowup:property-critical}.

\vspace{10pt}

\noindent {\bf Step 2 (Convergence of the blow-up sequence).} In this step, we show that there exists a sequence $\{x_{\beta}\}_{\beta} \subset \mathbb R^{2}$ such that
$$
\lim_{\beta \to 0}\varepsilon_{\beta} u_{\beta}(\varepsilon_{\beta} \cdot + x_{\beta}) = Q_{0}
$$
strongly in $H^{1}(\mathbb R^{2})$ where, by \eqref{blowup:property-critical},
\begin{equation}\label{length:blowup-critical}
\varepsilon_{\beta} := \|\nabla u_{\beta}\|_{L^{2}}^{-1} \xrightarrow{\beta \to 0} 0.
\end{equation}
Let's denote $w_{\beta}(x):= \varepsilon_{\beta} u_{\beta}(\varepsilon_{\beta} x)$. Then we have
\begin{equation}\label{cv:blow-up-critical-0}
\|w_{\beta}\|_{L^{2}} = 1 = \|\nabla w_{\beta}\|_{L^{2}}.
\end{equation}
We rewrite the average-field-Pauli energy as follows
\begin{equation}\label{energy:blow-up-critical}
E^{\rm afP}_{\beta,g_{*}(0)} = \varepsilon_{\beta}^{-2}\int_{\mathbb R^{2}} \left[\left| (-{\rm i}\nabla + \beta\mathbf{A}[|w_{\beta}|^{2}]) w_{\beta}\right|^{2} - g_{*}(0)|w_{\beta}|^{4}\right] + \varepsilon_{\beta}^{s}\int_{\mathbb R^{2}}V|w_{\beta}|^{2}.
\end{equation}
Multiplying both side of \eqref{energy:blow-up-critical} by $\varepsilon_{\beta}^{2}$, dropping the nonnegative trapping term and using the energy upper bound in \eqref{blowup:energy-afP}, we get
\begin{equation}\label{cv:blow-up-critical-1}
\lim_{\beta \to 0}\int_{\mathbb R^{2}} \left[\left| (-{\rm i}\nabla + \beta\mathbf{A}[|w_{\beta}|^{2}]) w_{\beta}\right|^{2} - g_{*}(0)|w_{\beta}|^{4}\right] = 0.
\end{equation}
Since, by \eqref{ineq:hardy},
$$
(1-C\beta)^{2} \int_{\mathbb R^{2}} |\nabla w_{\beta}|^{2} \leq  \int_{\mathbb R^{2}} \left| (-{\rm i}\nabla + \beta\mathbf{A}[|w_{\beta}|^{2}]) w_{\beta}\right|^{2} \leq (1+C\beta)^{2} \int_{\mathbb R^{2}} |\nabla w_{\beta}|^{2},
$$
we deduce from \eqref{cv:blow-up-critical-1} that
\begin{equation}\label{cv:blow-up-critical-2}
\lim_{\beta \to 0} \int_{\mathbb R^{2}} \left[\left|\nabla  w_{\beta}\right|^{2} - g_{*}(0)|w_{\beta}|^{4}\right] = 0.
\end{equation}
Looking back at \eqref{cv:blow-up-critical-0}, since $\{w_{\beta}\}_{\beta}$ is a bounded sequence in $H^{1}(\mathbb R^{2})$, we have, up to translation and extracting a subsequence, $w_{\beta} \to w_{0}$ weakly in $H^{1}(\mathbb R^{2})$ and almost everywhere in $\mathbb R^{2}$. We show that this is actually strong convergence in $H^{1}(\mathbb R^{2})$.

We first note that $w_{0} \not\equiv 0$. Otherwise, we must have that $w_{\beta} \to 0$ strongly in $L^{p}(\mathbb R^{2})$ for all $2<p<\infty$ (see \cite{Lions-84b}). This yields a contradiction, by \eqref{cv:blow-up-critical-2} and \eqref{cv:blow-up-critical-0}. On the other hand, if $0 < \|w_{0}\|_{L^{2}}^{2} < 1$ then it follows from \eqref{cv:blow-up-critical-2} and Brezis--Lieb lemma \cite{BreLie-83} that
\begin{equation}\label{cv:blow-up-critical-dichotomy-0}
0 = \int_{\mathbb R^{2}} \left[|\nabla w_{0}|^{2} - \frac{g_{*}(0)}{2}|w_{0}|^{4}\right] + \int_{\mathbb R^{2}} \left[|\nabla (w_{\beta} - w_{0})|^{2} - \frac{g_{*}(0)}{2}|w_{\beta} - w_{0}|^{4}\right] = \tilde{I}_{1} + \tilde{I}_{2}.
\end{equation}
By \eqref{ineq:gn}, we have
\begin{equation}\label{cv:blow-up-critical-dichotomy-1}
\tilde{I}_{1} \geq \left(\frac{1}{\int_{\mathbb R^{2}}|w_{0}|^{2}}-1\right) \frac{g_{*}(0)}{2} \int_{\mathbb R^{2}}|w_{0}|^{4} > 0.
\end{equation}
Furthermore,
\begin{equation}\label{cv:blow-up-critical-dichotomy-2}
\tilde{I}_{2} \geq \lim_{\beta \to 0}\int_{\mathbb R^{2}} \left(\frac{1}{\int_{\mathbb R^{2}}|w_{\beta} - w_{0}|^{2}}-1\right) \frac{g_{*}(0)}{2} \int_{\mathbb R^{2}}|w_{\beta} - w_{0}|^{4} \geq 0.
\end{equation}
Combining \eqref{cv:blow-up-critical-dichotomy-0}, \eqref{cv:blow-up-critical-dichotomy-1} and \eqref{cv:blow-up-critical-dichotomy-2} we obtain a contraction. Therefore, we must have $\|w_{0}\|_{L^{2}}^{2} = 1$. This yields that, by Brezis--Lieb lemma \cite{BreLie-83}, up to translation and extracting a subsequence, $w_{\beta} \to w_{0}$ strongly in $L^{2}(\mathbb R^{2})$. In fact, this strong convergence holds in $L^{p}(\mathbb R^{2})$ for any $2 \leq p < \infty$, by interpolation. We then deduce from Fatou's lemma and \eqref{cv:blow-up-critical-2} that
$$
\frac{g_{*}(0)}{2}\int_{\mathbb R^{2}}|w_{0}|^{4} = \lim_{\beta \to 0} \frac{g_{*}(0)}{2}\int_{\mathbb R^{2}}|w_{\beta}|^{4} = \lim_{\beta \to 0} \int_{\mathbb R^{2}} \left|\nabla w_{\beta}\right|^{2} \geq \int_{\mathbb R^{2}} \left|\nabla w_{0}\right|^{2}.
$$
The above yields that we must have $w_{\beta}\to w_{0}$ strongly in $H^{1}(\mathbb R^{2})$, by \eqref{ineq:gn}. Moreover, $w_{0}$ is an optimizer of the standard Gagliardo--Nirenberg inequality \eqref{ineq:gn}. By the uniqueness (up to translations and dilations) of optimizers for  \eqref{ineq:gn}, there exist $\lambda>0$ and $x_{0} \in \mathbb R^{2}$ such that $w_{0}(x) = \lambda Q_{0}(\lambda(x+x_{0}))$. Since 
$\|\nabla w_{0}\|_{L^{2}}^{2} = 1$, we must have $\lambda=1$. Again, by the uniqueness of $Q_{0}$, we conclude that passing to a subsequence is unnecessary.

\vspace{10pt}

\noindent {\bf Step 3. Smallness of the imaginary part.} We write
$$
w_{\beta}(x) = q_{\beta}(x) + {\rm i} r_{\beta}(x)
$$
with $q_{\beta}$ and $r_{\beta}$ the real and imaginary parts of $w_{\beta}$, respectively. Then $q_{\beta} \to Q_{0}$ and $r_{\beta} \to 0$ strongly in $H^{1}(\mathbb R^{2})$ since $w_{\beta} \to Q_{0}$ strongly in $H^{1}(\mathbb R^{2})$. Furthermore, we choose the phase of $w_{\beta}$ such that $w_{\beta}$ is the closest to its limit
$$
\left\|w_{\beta}-Q_{0}\right\|_{L^{2}}=\min_{\theta \in [0,2\pi)} \left\|e^{{\rm i} \theta} v_{\theta}-Q_{0}\right\|_{L^{2}} .
$$
This gives the following orthogonality condition
\begin{equation} \label{ortho-cond}
\int_{\mathbb R^{2}} Q_{0}r_{\beta} =0.
\end{equation}

We observe that, by using the integration by part and the fact that $\nabla \cdot \nabla^{\perp} = 0$,
\begin{align*}
\int_{\mathbb R^{2}} \mathbf{A}[|w_{\beta}|^{2}] \cdot \mathbf{J}[w_{\beta}] = \int_{\mathbb R^{2}} \mathbf{A}[|w_{\beta}|^{2}] \left(q_{\beta} \nabla r_{\beta} - r_{\beta} \nabla q_{\beta}\right) & = 2\int_{\mathbb R^{2}} \mathbf{A}[|w_{\beta}|^{2}] q_{\beta} \nabla r_{\beta} \\
& = -2\int_{\mathbb R^{2}} \mathbf{A}[|w_{\beta}|^{2}] r_{\beta} \nabla q_{\beta}.
\end{align*}
This yields that
$$
\left|\int_{\mathbb R^{2}} \mathbf{A}[|w_{\beta}|^{2}] \cdot \mathbf{J}[w_{\beta}]\right| \leq 2\|\mathbf{A}[|w_{\beta}|^{2}] q_{\beta}\|_{L^{2}} \|\nabla r_{\beta}\|_{L^{2}} \leq C \|\nabla r_{\beta}\|_{L^{2}}.
$$
Here we have used the fact that $\mathbf{A}[|w_{\beta}|^{2}] q_{\beta}$ is bounded uniformly in $L^{2}(\mathbb R^{2})$ since $w_{\beta} \to Q_{0}$ strongly in $H^{1}(\mathbb R^{2})$. We deduce from the above that
\begin{align*}
\varepsilon_{\beta}^{2} E^{\rm afP}_{\beta,g_{*}(0)} \geq \int_{\mathbb R^{2}}|\nabla q_{\beta}|^{2} +|\nabla r_{\beta}|^{2} -\frac{g_{*}(0)}{2} (q_{\beta}^{4} + r_{\beta}^{4} + 2q_{\beta}^{2} r_{\beta}^{2}) - C \beta \|\nabla r_{\beta}\|_{L^{2}}.
\end{align*}
One one hand, we have
\begin{align*}
\frac{g_{*}(0)}{2} \int_{\mathbb R^{2}}(r_{\beta}^{4} + 2q_{\beta}^{2} r_{\beta}^{2}) \leq g_{*}(0) \int_{\mathbb R^{2}}|w_{\beta}|^{2} r_{\beta}^{2} & = \int_{\mathbb R^{2}}Q^{2} r_{\beta}^{2} + g_{*}(0) \int(|w_{\beta}|^{2}-Q_{0}^{2}) r_{\beta}^{2}
\\ & = \int_{\mathbb R^{2}}Q^{2} r_{\beta}^{2} + o(1)\|r_{\beta}\|^{2}_{H^{1}}.
\end{align*}
Here we have used Cauchy--Schwarz inequality to obtain that
\begin{align*}
\left|\int_{\mathbb R^{2}}(|w_{\beta}|^{2}-Q_{0}^{2})r_{\beta}^{2}\right| \leq \||w_{\beta}|^{2}-Q_{0}^{2}\|_{L^{2}} \|r_{\beta}\|^{2}_{L^{4}} & \leq C\||w_{\beta}|^{2}-Q_{0}^{2}\|_{L^{2}} \|r_{\beta}\|^{2}_{H^{1}} \\
& \leq \||w_{\beta}|+Q_{0}\|_{L^{4}} \||w_{\beta}|-Q_{0}\|_{L^{4}} \|r_{\beta}\|^{2}_{H^{1}}
\end{align*}
as well as the strong convergence $|w_{\beta}| \to Q_{0}$ in $H^{1}(\mathbb R^{2})$. On the other hand, by the standard Gagliardo--Nirenberg inequality \eqref{ineq:gn}, the strong convergence $q_{\beta} \to Q_{0}$ in $H^{1}(\mathbb R^{2})$ and the fact that $\|\nabla Q_{0}\|^{2}_{L^{2}}=1$ as well as $\|q_{\beta}\|^{2}_{L^{2}} + \|r_{\beta}\|^{2}_{L^{2}} = \|w_{\beta}\|^{2}_{L^{2}} =1$, we have
$$
\int_{\mathbb R^{2}}|\nabla q_{\beta}|^{2} - \frac{g_{*}(0)}{2} q_{\beta}^{4}  \geq \|\nabla q_{\beta}\|^{2}_{L^{2}}(1-\|q_{\beta}\|^{2}_{L^{2}}) = (1+o(1)_{\beta \to 0}) \|r_{\beta}\|^{2}_{L^{2}}.
$$
Collecting all the above estimates, we get
\begin{align}\label{est-rn}
\varepsilon_{\beta}^{2} E^{\rm afP}_{\beta,g_{*}(0)} &\geq \int_{\mathbb R^{2}} \left[|\nabla r_{\beta}|^{2} - Q^{2} r_{\beta}^{2} +r_{\beta}^{2}\right] + o(1) \|r_{\beta}\|^{2}_{H^{1}} - C\beta\|\nabla r_{\beta}\|_{L^{2}} \nonumber \\
&= \langle r_{\beta}, \mathcal{L} r_{\beta}\rangle + \|r_{\beta}\|^{2}_{H^{1}} o(1)_{\beta \to 0} - C\beta\|\nabla r_{\beta}\|_{L^{2}},
\end{align}
where we denoted $\mathcal{L} := -\Delta - Q^{2} +1$. 

We now use the non-degeneracy property of $Q$ to deal with the term of $\mathcal{L}$ in the above. It is well-known (see \cite[Theorem 11.8 and Corrollary 11.9]{LieLos-01}) that $Q$ is the first eigenfunction of $\mathcal{L}$ and the corresponding eigenvalue $\lambda_{1} = 0$ is non-degenerate. In particular, we have
$$
\langle \mathcal{L} u, u\rangle \geq \lambda_2 \|u\|^{2}_{L^{2}}
$$
for all $u$ orthogonal to $Q$, where $\lambda_2>0$ is the second eigenvalue of $\mathcal{L}$. This together with the fact that
$$
\langle \mathcal{L} u, u\rangle \geq \|u\|^{2}_{H^{1}} -\|Q\|^{2}_{L^\infty} \|u\|^{2}_{L^{2}}
$$
yield the estimate
\begin{equation}\label{est-rn-0}
\langle \mathcal{L} u, u\rangle \geq C \|u\|^{2}_{H^{1}}
\end{equation}
for some constant $C>0$ and all $u$ orthogonal to $Q$. It follows from \eqref{est-rn}, \eqref{est-rn-0} and the orthogonality condition \eqref{ortho-cond} that
$$
\varepsilon_{\beta}^{2} E^{\rm afP}_{\beta,g_{*}(0)} \geq C \|r_{\beta}\|^{2}_{H^{1}} - C\beta\|\nabla r_{\beta}\|_{L^{2}}.
$$
This implies that
\begin{align} \label{est-rn-1}
C\|r_{\beta}\|^{2}_{H^{1}} \leq \varepsilon_{\beta}^{2} E^{\rm afP}_{\beta,g_{*}(0)} + \beta^{2}.
\end{align}
On the other hand, from \eqref{ineq:gn} and the energy upper bound in \eqref{blowup:energy-afP} we have
$$
C\beta^{\frac{2s}{s+2}} \geq E^{\rm afP}_{\beta,g_{*}(0)} = \mathcal{E}^{\rm afP}_{\beta,g_{*}(0)}[u_{\beta}] \geq \int_{\mathbb R^{2}}V|u_{\beta}|^{2} = \varepsilon_{\beta}^{s}\int_{\mathbb R^{2}}V|w_{\beta}|^{2}.
$$
We deduce from the above and the strong convergence $w_{\beta} \to Q_{0}$ in $H^{1}(\mathbb R^{2})$ that
\begin{align}\label{est-vareps-n}
\varepsilon_{\beta} \leq C \beta^{\frac{2}{s+2}}.
\end{align}
Putting together \eqref{est-rn-1}, \eqref{est-vareps-n} and the energy upper bound in \eqref{blowup:energy-afP}, we obtain
\begin{align}\label{est:imaginary}
\|r_{\beta}\|_{H^{1}} \leq C \beta.
\end{align}

\vspace{10pt}

\noindent {\bf Step 4. Identifying the blow-up limit.} By the definition of $\mathbf{A}$ in \eqref{potential:magnetic} and of $\mathbf{J}$ in \eqref{potential:current}, we have the following decomposition
\begin{equation}\label{est:cross-term-0}
\int_{\mathbb R^{2}} \mathbf{A}[|w_{\beta}|^{2}] \cdot \mathbf{J}[w_{\beta}] = \int_{\mathbb R^{2}} \mathbf{A}[|q_{\beta}|^{2}] q_{\beta} \nabla r_{\beta} + \int_{\mathbb R^{2}} \mathbf{A}[|r_{\beta}|^{2}] q_{\beta} \nabla r_{\beta}.
\end{equation}
It is worth noting that $w_{0}$, $Q_{0}$ as well as $w_{0}*Q_{0}$ are radial symmetric. Using integration by part, we have
$$
\int_{\mathbb R^{2}} \mathbf{A}[Q_{0}|^{2}] Q_{0} \nabla r_{\beta} = -\int_{\mathbb R^{2}} \mathbf{A}[Q_{0}|^{2}] r_{\beta} \nabla Q_{0} = 0.
$$
Therefore, one can estimate the first term on the right hand side of \eqref{est:cross-term-0} as follow
\begin{align}\label{est:cross-term-1}
\left|\int_{\mathbb R^{2}} \mathbf{A}[|q_{\beta}|^{2}] q_{\beta} \nabla r_{\beta}\right| & = \left|\int_{\mathbb R^{2}} \left(\mathbf{A}[|q_{\beta}|^{2}] q_{\beta} - \mathbf{A}[Q_{0}^{2}] Q_{0}\right) \nabla r_{\beta}\right| \nonumber \\
& \leq \|\mathbf{A}[|q_{\beta}|^{2}] q_{\beta} - \mathbf{A}[Q_{0}^{2}] Q_{0}\|_{L^{2}} \|\nabla r_{\beta}\|_{L^{2}} \nonumber \\
& = o(\beta)_{\beta \to 0}
\end{align}
where we have used \eqref{est:imaginary} and the strong convergence $q_{\beta} \to Q_{0}$ in $H^{1}(\mathbb R^{2})$. For the second term on the right hand side of \eqref{est:cross-term-0}, we use again \eqref{est:imaginary} and \eqref{ineq:hardy} to obtain that
\begin{align}\label{est:cross-term-2}
\left|\int_{\mathbb R^{2}} \mathbf{A}[|r_{\beta}|^{2}] q_{\beta} \nabla r_{\beta}\right| = \left|\int_{\mathbb R^{2}} \mathbf{A}[|r_{\beta}|^{2}] r_{\beta} \nabla q_{\beta}\right| & \leq \|\mathbf{A}[|r_{\beta}|^{2}] r_{\beta}\|_{L^{2}} \|\nabla q_{\beta}\|_{L^{2}} \nonumber \\
& \leq C\|r_{\beta}\|_{L^{2}}^{2} \|\nabla r_{\beta}\|_{L^{2}} \|\nabla q_{\beta}\|_{L^{2}} \nonumber \\
& \leq C\beta^{3}.
\end{align}
Combining \eqref{est:cross-term-0}, \eqref{est:cross-term-1} and \eqref{est:cross-term-2}, we arrive at 
\begin{equation}\label{est:cross-term}
\int_{\mathbb R^{2}} \mathbf{A}[|w_{\beta}|^{2}] \cdot \mathbf{J}[w_{\beta}] = o(\beta)_{\beta \to 0}.
\end{equation}

We now denote
$$
\kappa_{\beta} = \frac{\varepsilon_{\beta}}{\beta^{\frac{2}{s+2}}}
$$
which is bounded uniformly from above, by \eqref{est:cross-term-0}. Furthermore, it follows from \eqref{energy:blow-up-critical}, \eqref{ineq:gn} and the energy upper bound in \eqref{blowup:energy-afP} that
\begin{equation}\label{lower-bound:energy-critical}
\frac{s+2}{s}\beta^{\frac{2s}{s+2}} \mathcal{Q}_{s}^{\frac{2}{s+2}} \mathcal{A}_{0}^{\frac{s}{s+2}} \geq \varepsilon_{\beta}^{-2}\int_{\mathbb R^{2}} \left[2\beta\mathbf{A}[|w_{\beta}|^{2}] \cdot \mathbf{J}[w_{\beta}] + \beta^{2} |\mathbf{A}[|w_{\beta}|^{2}]|^{2}|w_{\beta}|^{2} \right] + \varepsilon_{\beta}^{s}\int_{\mathbb R^{2}}V|w_{\beta}|^{2}.
\end{equation}
Multiplying both sides of \eqref{lower-bound:energy-critical} by $\varepsilon_{\beta}^{2} {\beta}^{-2}$, dropping the nonnegative trapping term and using \eqref{est:cross-term}, we must have that $\kappa_{\beta}$ is bounded away from zero in the limit $\beta \to 0$. Hence
$$
\lim_{\beta \to 0}\kappa_{\beta} = \kappa_{0} > 0.
$$
Next, we determine $\kappa_{0}$ by calculating the matching energy lower bound in \eqref{blowup:energy-afP}. Multiplying both sides of \eqref{lower-bound:energy-critical} by $\beta^{-\frac{2s}{s+2}}$, using \eqref{est:cross-term} and taking the limit $\beta \to 0$, we get
\begin{align}\label{cv:minimizer-afP-nls-critical}
\frac{s+2}{s} \mathcal{Q}_{s}^{\frac{2}{s+2}} \mathcal{A}_{0}^{\frac{s}{s+2}} & \geq \kappa_{0}^{-2} \int_{\mathbb R^{2}}\mathbf{A}[|w_{0}|^{2}]|^{2}|w_{0}|^{2} + \kappa_{0}^{s}\int_{\mathbb R^{2}}|x|^{s}|w_{0}|^{2} \nonumber \\
& = \kappa_{0}^{-2} \mathcal{A}_{0} + \kappa_{0}^{s}\int_{\mathbb R^{2}}|x-x_{0}|^{2}|Q_{0}|^{2} \nonumber \\
& \geq \kappa_{0}^{-2} \mathcal{A}_{0} + \kappa_{0}^{s}\frac{2}{s}\mathcal{Q}_{s}.
\end{align}
Here we have used the rearrangement inequality as $Q_{0}$ is symmetric decreasing and $|x|^{s}$ is strictly symmetric increasing. It is elementary to check that
$$
\inf_{r>0}\left(r^{-2} \mathcal{A}_{0} + r^{s}\frac{2}{s}\mathcal{Q}_{s}\right) = \frac{s+2}{s} \mathcal{Q}_{s}^{\frac{2}{s+2}} \mathcal{A}_{0}^{\frac{s}{s+2}}
$$
and the equality is attained at
$$
r_{0} = \mathcal{Q}_{s}^{-\frac{1}{s+2}} \mathcal{A}_{0}^{\frac{1}{s+2}}.
$$
Therefore, we must have equality in \eqref{cv:minimizer-afP-nls-critical} and it yields that $\kappa_{0}=r_{0}$ and $x_{0}=0$. This conclude the energy lower bound in \eqref{blowup:energy-afP} as well as the convergence of blow-up sequence in \eqref{blowup:minimizer-afP}.

\subsubsection{The super-critical case}\label{sec:supercritical-afP}

The proof of \Cref{thm:collapse-afP} in the super-critical regime borrows arguments in \Cref{sec:subcritical-afP,sec:critical-afP}. For reader's convenience, we give some important steps. As usual, by the variational principle, we have
\begin{equation}\label{cv:energy-afP-supercritical-upper-bound}
E_{\beta,g_{\beta}}^{\rm afP} \leq \mathcal{E}^{\rm afP}_{\beta,g_{\beta}}[\ell^{-1} Q_{0}(\ell^{-1}\cdot)] = \ell^{-2}\left(1-\frac{g_{\beta}}{g_{*}(0)} + \beta^{2}\mathcal{A}_{0}\right) + \ell^{s}\frac{2}{s}\mathcal{Q}_{s}
\end{equation}
for any $\ell>0$. Under the assumption $g_{\beta} = g_{*}(0) + \tau_{0}\beta^{2}$ with $0 < \tau_{0} < \tau_{*} \leq g_{*}(0)\mathcal{A}_{0}$, by optimizing over all $\ell>0$ of \eqref{cv:energy-afP-supercritical-upper-bound}, we obtain the energy upper bound in \eqref{blowup:energy-afP}. As usual, we prove the energy lower bound by justifying the convergence of minimizers in \eqref{blowup:minimizer-afP}.

Let's denote $w_{\beta}=\ell_{\beta}u_{\beta}(\ell_{\beta}\cdot)$ where $\ell_{\beta}$ is given by \eqref{length:supercritical}. Then $\|w_{\beta}\|_{L^{2}} = \|u_{\beta}\|_{L^{2}} = 1$ and, by \eqref{ineq:gn-magnetic},
\begin{align}
\mathcal{E}^{\rm afP}_{\beta,g_{\beta}}[u_{\beta}] & = 
\ell_{\beta}^{-2}\int_{\mathbb R^{2}} \left[ \left| (-{\rm i}\nabla + \beta\mathbf{A}[|w_{\beta}|^{2}]) w_{\beta}\right|^{2} - g_{\beta} |w_{\beta}|^{4} \right] + \ell_{\beta}^{s}\int_{\mathbb R^{2}} V|w_{\beta}|^{2} \label{cv:minimizer-afP-nls-supercritical-1} \\
& \geq \ell_{\beta}^{-2}\left(1-\frac{g_{\beta}}{g_{*}(\beta)}\right)\int_{\mathbb R^{2}} \left|(-{\rm i}\nabla + \beta\mathbf{A}[|w_{\beta}|^{2}])w_{\beta}\right|^{2} + \ell_{\beta}^{s}\int_{\mathbb R^{2}} V|w_{\beta}|^{2} \label{cv:minimizer-afP-nls-supercritical-2}
\end{align}
It follows from \eqref{cv:minimizer-afP-nls-supercritical-2}, \eqref{ineq:gn-weak-field} and the energy upper bound in \eqref{blowup:energy-afP} that
\begin{align}\label{cv:minimizer-afP-supercritical-boundedness}
\frac{s+2}{s}\left(1+o(1)_{\beta \to 0}\right) \geq{} & \frac{\tau_{*}-\tau_{0}}{g_{*}(\beta)}\left(-\frac{\tau_{0}}{g_{*}(0)}+\mathcal{A}_{0}\right)^{-1}\int_{\mathbb R^{2}} \left| (-{\rm i}\nabla + \beta\mathbf{A}[|w_{\beta}|^{2}]) w_{\beta}\right|^{2} \nonumber \\
& + \mathcal{Q}_{s}^{-1}\int_{\mathbb R^{2}} V|w_{\beta}|^{2}.
\end{align}
We deduce from \eqref{cv:minimizer-afP-supercritical-boundedness}, \eqref{ineq:hardy} and \eqref{ineq:diamagnetic} that $\{w_{\beta}\}$ is uniformly bounded in $H^{1} \cap L_{V}^{2} (\mathbb R^{2})$. Similarly to the arguments in \Cref{sec:subcritical-afP}, we deduce that, after extracting a subsequence, $w_{\beta} \to w_{0}$ weakly in $H^{1}(\mathbb R^{2})$, almost everywhere in $\mathbb R^{2}$, and strongly in $L^{p}(\mathbb R^{2})$ for any $2 \leq p < \infty$. Furthermore,
$$
w_{0}(x) = t_{0}Q_{0}(t_{0}x + x_{0})
$$
for some $t_{0}>0$, $x_{0}\in\mathbb R^{2}$ and $Q_{0}$ being the (unique) real-valued solution of \eqref{eq:nls}. Finally, we again need to prove that $w_{0} \equiv Q_{0}$, by verifying that $t_{0}=1$ and $x_{0} \equiv 0$. This requires to show that the cross term \eqref{est:cross-term-0} is of small order to the leading order term in the collapse regime. For this purpose, we rewrite the energy as follows
\begin{align*}
& \varepsilon_{\beta}^{-2}\int_{\mathbb R^{2}} \left[\left| (-{\rm i}\nabla + \beta\mathbf{A}[|w_{\beta}|^{2}]) w_{\beta}\right|^{2} - g_{*}(0)|w_{\beta}|^{4}\right] + \varepsilon_{\beta}^{s}\int_{\mathbb R^{2}}V|w_{\beta}|^{2} \\
& = E^{\rm afP}_{\beta,g_{\beta}} + \frac{g_{\beta} - g_{*}(0)}{2} \ell_{\beta}^{-2}\int_{\mathbb R^{2}} |w_{\beta}|^{4} \\
& = \left( C + o(1)_{\beta \to 0} \right) E^{\rm afP}_{\beta,g_{\beta}}
\end{align*}
where we have used the strong convergence $w_{\beta} \to w_{0}$ in $L^{p}(\mathbb R^{2})$, for any $2 \leq p < \infty$. At this stage, the arguments are exactly the same as in \Cref{sec:critical-afP} where we identified directly $\varepsilon_{\beta}$ in \eqref{length:blowup-critical} with $\ell_{\beta}$ in \eqref{length:critical}. We omit details for brevity.

\section{Condensation and collapse in the many-body theory}\label{sec:many-body}

The goal of this section is to prove \Cref{thm:many-body}. It is worth noting that the (regularized) average-field-Hartree theory interpolates between the many-body and the average-field-Pauli theories. The strategy of the proof is to compare the (regularized) quantum energy $E^{\rm afQM}_{N^{\nu},R,\beta,g}$ in \eqref{energy:qm} and the (regularized) average-field-Hartree energy $E^{\rm afH}_{N^{\nu},R,\beta,g}$ in \eqref{energy:afH}. Therefore, it is necessary to make a relation between the average-field-Hartree and the average-field-Pauli theories. This might not difficult in this sub-critical regime where $0<g<g_{*}(0)$ since the (regularized) magnetic potential can be neglected. However, this is non-trivial in the critical and super-critical.

\subsection{Condensation and collapse in the average-field-Hartree theory}

In order to study the collapse phenomenon in the average-field-Hartree theory, we need the following lemma which relies on the convergence (rate) of the two-body interactions and the (regularized) kinetic energy with magnetic field.

\begin{lemma}\label{lem:afH-afP-interaction}
Let $\beta,R,\nu>0$ and $W$ be satisfying \eqref{condition:potential-two-body}. Then for any $u\in H^{1}(\mathbb R^{2})$, we have
\begin{equation}\label{cv-ineq:afH-afP-2body}
0 \leq \|u\|_{L^{4}}^{4} - \iint_{\mathbb R^{4}} W_{N^\nu}(x-y)|u(x)|^{2}|u(y)|^{2}{\rm d}x{\rm d}y \leq \|u\|_{H^{1}}^{4} o(1)_{N\to\infty}.
\end{equation}
Assume in addition that $xW \in L^1(\mathbb R^{2})$ then
\begin{equation}\label{cv-rate:afH-afP-2body}
0 \leq \|u\|_{L^{4}}^{4} - \iint_{\mathbb R^{4}} W_{N^\nu}(x-y)|u(x)|^{2}|u(y)|^{2}{\rm d}x{\rm d}y \leq CN^{-\nu} \|xW\|_{L^{1}} \|\nabla u\|_{L^{2}}^{3}.
\end{equation}
Furthermore, the non-interacting functional $\mathcal{E}^{\rm afP}_{R,\beta,0}$ converges pointwise to $\mathcal{E}^{\rm afP}_{\beta,0}$ as $R \to 0$ and we have
\begin{equation}\label{cv:afH-afP-curvature}
\left| \mathcal{E}^{\rm afP}_{R,\beta,0}[u] - \mathcal{E}^{\rm afP}_{\beta,0}[u] \right| \leq C\beta R \|\nabla u\|_{L^{2}}^{3}.
\end{equation}
\end{lemma}

\begin{proof}
The proof of \eqref{cv-ineq:afH-afP-2body} and \eqref{cv-rate:afH-afP-2body} can be found in \cite{LewNamRou-17,LewNamRou-17-proc}. On the other hand, \eqref{cv:afH-afP-curvature} follows from \cite[Proposition 3.9]{LunRou-15}.
\end{proof}

\begin{theorem}[Condensation and collapse of the average-field-Hartree ground states]\label{thm:collapse-afH}
Let $\eta,\nu > 0$, $R = R_{N} \sim N^{-\eta}$ and $W$ be satisfying \eqref{condition:potential-two-body}.

\begin{enumerate}[label=(\roman*)]
\item Let $\beta,g > 0$ be fixed with $g < g_{*}(\beta)$. Assume that $\nu < \eta$ if $g \geq g_{*}(0)$. Then we have
\begin{equation}\label{cv:afH-afP-energy}
\lim_{N \to \infty}E^{\rm afH}_{N^{\nu},R_{N},\beta,g} = E^{\rm afP}_{\beta,g}.
\end{equation}
Moreover, if $\{u_{N}\}_{N}$ is a sequence of (approximate) ground states of $E^{\rm afH}_{N^{\nu},R_{N},\beta,g}$, i.e., $\mathcal{E}^{\rm afH}_{N^{\nu},R_{N},\beta,g}[u_{N}] = E^{\rm afH}_{N^{\nu},R_{N},\beta,g}(1 + o(1)_{N\to+\infty})$, then there exists a minimizer $u_{0}$ of $E^{\rm afP}_{\beta,g}$ such that, up to extracting a subsequence,
\begin{equation}\label{cv:afH-afP-ground-state}
\lim_{N \to \infty}\|u_{N} - u_{0}\|_{H^{1}} = 0.
\end{equation}

\item Let $Q_{0}$, $\mathcal{Q}_{s}$, $\mathcal{A}_{0}$, $\{\beta_{N}\}_{N}$ with $\beta_{N} \searrow 0$, $\{g_{N}\}_{N}$ with $g_{N} \to g_{*}(0)$, and $\{\ell_{N}\}_{N}$ be as in \Cref{thm:collapse-afP}. Assume that $xW \in L^1(\mathbb R^{2})$ and $\ell_{N} \sim N^{-\sigma}$ with
\begin{equation}\label{collapse:speed-hartree}
0 < \sigma < \min\left\{\frac{\nu}{s+3},\frac{2\eta}{s+4}\right\}.
\end{equation}
Assume further that $\nu < \eta$ in the critical case and additionally $\sigma < \frac{2(\eta-\nu)}{s+2}$ in the super-critical case. Then we have
\begin{equation}\label{blowup:energy-afH}
E^{\rm afH}_{N^{\nu},R_{N},\beta_{N},g_{N}} = E^{\rm afP}_{\beta_{N},g_{N}} (1+o(1)_{N\to\infty})  = \ell_{N}^{s} \left(\frac{s+2}{s}\mathcal{Q}_{s} + o(1)_{N\to\infty}\right).
\end{equation}
Moreover, if $\{u_{N}\}_{N}$ is a sequence of (approximate) ground states of $E^{\rm afH}_{N^{\nu},R_{N},\beta_{N},g_{N}}$, i.e., $\mathcal{E}^{\rm afH}_{N^{\nu},R_{N},\beta_{N},g_{N}}[u_{N}] = E^{\rm afH}_{N^{\nu},R_{N},\beta_{N},g_{N}}(1 + o(1)_{N\to+\infty})$, then we have
\begin{equation}\label{blowup:minimizer-afH}
\lim_{N \to \infty} \|\ell_{N} u_{N} ( \ell_{N} \cdot ) - Q_{0}\|_{H^{1}} = 0.
\end{equation}
\end{enumerate}
\end{theorem}

\begin{proof}
First, we prove the convergence of Hartree energy in \eqref{cv:afH-afP-energy} as well as of its ground states in \eqref{cv:afH-afP-ground-state}. Let $u_{0} \in H^{1}(\mathbb R^{2})$ be a minimizer of $E^{\rm afP}_{\beta,g}$. It follows immediately from the variational principle together with \eqref{cv-ineq:afH-afP-2body} and \eqref{cv:afH-afP-curvature} that
$$
\lim_{N \to \infty}E^{\rm afH}_{N^{\nu},R_{N},\beta,g} \leq \lim_{N \to \infty} \mathcal{E}^{\rm afH}_{N^{\nu},R_{N},\beta,g}[u_{0}] \leq \mathcal{E}^{\rm afP}_{\beta,g}[u_{0}] = \mathcal{E}^{\rm afP}_{\beta,g}[u_{0}] = E^{\rm afP}_{\beta,g}.
$$
This is the desired energy upper bound in \eqref{cv:afH-afP-energy}.

We prove the matching lower bound \eqref{cv:afH-afP-energy} by verifying the convergence of ground states in \eqref{cv:afH-afP-ground-state}. Let $\{u_{N}\}_{N}$ be a sequence of (approximate) ground state of $E^{\rm afH}_{N^{\nu},R_{N},\beta,g}$. In the stable regime where $0<g<g_{*}(\beta)$, we claim that $\{u_{N}\}_{N}$ is bounded uniformly in $H^{1}\cap L_{V}^{2}(\mathbb R^{2})$. In fact, this is obvious in the case $0<g<g_{*}(0)$, by \eqref{ineq:gn} and \eqref{ineq:diamagnetic}. We prove that this still holds true when $g \geq g_{*}(0)$. We assume on the contrary that $\|(-{\rm i}\nabla + \beta\mathbf{A}_{R_{N}}[|u_{N}|^{2}])u_{N}\|_{L^{2}} \to \infty$ as $N \to \infty$. We define $\widetilde{u_{N}} = \varepsilon_{N} u_{N}(\varepsilon_{N} \cdot)$ with 
\begin{equation}\label{boundedness:afH-ground-state-length}
\varepsilon_{N}^{-2} := \int_{\mathbb R^{2}}\left|(-{\rm i}\nabla + \beta\mathbf{A}_{R_{N}}[|u_{N}|^{2}])u_{N}\right|^{2}.
\end{equation}
Noticing that
\begin{equation}\label{eq:afH-ground-state-renormalized}
\int_{\mathbb R^{2}}\left|\widetilde{u_{N}}\right|^{2} = 1 = \int_{\mathbb R^{2}}\left|(-{\rm i}\nabla + \beta\mathbf{A}_{R_{N}\varepsilon_{N}^{-1}}[|\widetilde{u_{N}}|^{2}])\widetilde{u_{N}}\right|^{2}.
\end{equation}
Using \eqref{ineq:hardy}, \eqref{ineq:diamagnetic} and \eqref{eq:afH-ground-state-renormalized} with $\mathbf{A}$ replaced by $\mathbf{A}_{R_{N}\varepsilon_{N}^{-1}}$, we deduce that $\widetilde{u_{N}}$ is bounded uniformly in $H^{1}(\mathbb R^{2})$. We then use the first inequality in \eqref{cv-ineq:afH-afP-2body} to obtain that
\begin{equation}\label{boundedness:afH-ground-state}
\mathcal{E}^{\rm afH}_{N^{\nu},R_{N},\beta,g}[u_{N}] \geq \varepsilon_{N}^{-2}\int_{\mathbb R^{2}}\left[\left|(-{\rm i}\nabla + \beta\mathbf{A}_{R_{N}\varepsilon_{N}^{-1}}[|\widetilde{u_{N}}|^{2}])\widetilde{u_{N}}\right|^{2} - \frac{g}{2}\left|\widetilde{u_{N}}\right|^{4}\right] + \varepsilon_{N}^{s}\int_{\mathbb R^{2}}V|\widetilde{u_{N}}|^{2}.
\end{equation}
Multiplying both side of \eqref{boundedness:afH-ground-state} with $\varepsilon_{N}^{2} \to 0$ as $N\to\infty$, dropping the nonnegative external potential, using \eqref{eq:afH-ground-state-renormalized} and the energy upper bound in \eqref{cv:afH-afP-energy}, we deduce that
\begin{equation}\label{conv:afH-afP-ground-state-0}
\lim_{N \to \infty}\int_{\mathbb R^{2}}\left|\widetilde{u_{N}}\right|^{4} \geq \frac{2}{g} > 0.
\end{equation}
This, however, is impossible. Indeed, we first note that $\varepsilon_{N}^{-1} \leq CN^{\nu}$, by \eqref{condition:potential-two-body} and \eqref{boundedness:afH-ground-state-length}. This yields that, under the assumption $\nu<\eta$,
$$
R_{N}\varepsilon_{N}^{-1} \leq N^{-\eta+\nu} \xrightarrow{N \to \infty} 0.
$$
We then use \eqref{ineq:gn-magnetic}, \eqref{cv:afH-afP-curvature}, \eqref{eq:afH-ground-state-renormalized} to deduce that
\begin{align*}
\lim_{N \to \infty}\int_{\mathbb R^{2}}\left|\widetilde{u_{N}}\right|^{4} & \leq \frac{2}{g_{*}(\beta)}\lim_{N \to \infty} \int_{\mathbb R^{2}}\left|(-{\rm i}\nabla + \beta\mathbf{A}[|\widetilde{u_{N}}|^{2}])\widetilde{u_{N}}\right|^{2} \\
& \leq \frac{2}{g_{*}(\beta)}\lim_{N \to \infty} \int_{\mathbb R^{2}}\left|(-{\rm i}\nabla + \beta\mathbf{A}_{R_{N}\varepsilon_{N}^{-1}}[|\widetilde{u_{N}}|^{2}])\widetilde{u_{N}}\right|^{2} \\
& = \frac{2}{g_{*}(\beta)} < \frac{2}{g}.
\end{align*}
Here we have used the fact that $g<g_{*}(\beta)$. The above contradicts \eqref{conv:afH-afP-ground-state-0}. Therefore, we must have that $\|(-{\rm i}\nabla + \beta\mathbf{A}_{R_{N}}[|u_{N}|^{2}])u_{N}\|_{L^{2}}$ is bounded uniformly. This in turn implies that $\{u_{N}\}_{N}$ is bounded uniformly in $H^{1}(\mathbb R^{2})$ as well, by using again \eqref{eq:afH-ground-state-renormalized}, \eqref{ineq:hardy} and \eqref{ineq:diamagnetic} with $\mathbf{A}$ replaced by $\mathbf{A}_{R_{N}}$. Then we use again \eqref{cv:afH-afP-curvature} and the first inequality in \eqref{cv-ineq:afH-afP-2body} to obtain that
$$
o(1)_{N\to\infty} + \mathcal{E}^{\rm afH}_{N^{\nu},R_{N},\beta,g}[u_{N}] \geq \mathcal{E}^{\rm afP}_{\beta,0}[u_{N}] - \frac{g}{2}\|u_{N}\|_{L^{4}}^{4} + \int_{\mathbb R^{2}}V|u_{N}|^{2}.
$$
It then follows from \eqref{ineq:gn-magnetic} and the energy upper bound in \eqref{cv:afH-afP-energy} that $\{u_{N}\}_{N}$ is bounded uniformly in $L_{V}^{2}(\mathbb R^{2})$. Therefore, up to a subsequence, $u_{N}\to u_{0}$ weakly in $H^{1}(\mathbb R^{2})$, almost everywhere in $\mathbb R^{2}$, and strongly in $L^r(\mathbb R^{2})$ for $2\leq r<\infty$. In particular, we have $\|u_{0}\|_{L^{2}}=1$ and
$$
\lim_{N \to \infty}\mathcal{E}^{\rm afH}_{N^{\nu},R,\beta,g}[u_{N}] \geq \mathcal{E}^{\rm afP}_{\beta,g}[u_{0}] \geq E^{\rm afP}_{N^{\nu},\beta,g},
$$
by \eqref{cv-ineq:afH-afP-2body}, \eqref{cv:afH-afP-curvature} and the weak lower semicontinuity (see \cite[Proposition 3.7]{LunRou-15}). The above yields the desired energy lower bound in \eqref{cv:afH-afP-energy} as well as the convergence of ground states in \eqref{cv:afH-afP-ground-state}.

Now we complete the proof of \Cref{thm:collapse-afH} by establishing the asymptotic behavior of the (regularized) average-field-Hartree energy and its ground states in the collapse regime. By the variational principle and the second inequality in \eqref{cv-rate:afH-afP-2body}, we have
\begin{align}\label{blowup:afH-upper-bound}
E^{\rm afH}_{N^{\nu},R_{N},\beta_N,g_N} & \leq \mathcal{E}^{\rm afH}_{N^{\nu},R_{N},\beta_N,g_N}[\ell_{N}^{-1}Q_{0}(\ell_{N}^{-1}\cdot) \nonumber\\
& \leq \mathcal{E}^{\rm afP}_{\beta_{N},g_{N}}[\ell_{N}^{-1}Q_{0}(\ell_{N}^{-1}\cdot)] + CN^{-\nu} \ell_{N}^{-3} + C\beta_{N}R_{N}\ell_{N}^{-3} \nonumber \\
& = \ell_{N}^{s}\left(\frac{s+2}{s}\mathcal{Q}_{s} + CN^{-\nu} \ell_{N}^{-s-3} + C\beta_{N}N^{-\eta}\ell_{N}^{-s-3} + o(1)_{N\to\infty} \right).
\end{align}
Since $\ell_{N} \sim N^{-\sigma}$ with $\sigma > 0$, the error term $N^{-\nu} \ell_{N}^{-s-3}$ is negligible when $\sigma<\frac{\nu}{s+3}$. Furthermore, the error term $\beta_{N}N^{-\eta}\ell_{N}^{-s-3}$ is also negligible when $\sigma < \frac{2\eta}{s+4}$ in the all three cases of \Cref{thm:collapse-afP}. This gives the energy upper bound in \eqref{blowup:energy-afH}. 

We establish the matching energy lower bound in \eqref{blowup:energy-afH} by verifying the convergence of the blowup sequence in \eqref{blowup:minimizer-afH}. Let $u_{N}$ be a sequence of (approximate) ground state of $E^{\rm afH}_{N^{\nu},R_{N},\beta_N,g_N}$ and $w_{N}:= \ell_{N}u_{N}(\ell_{N}\cdot)$. Then, $\|w_{N}\|_{L^{2}} = \|u_{N}\|_{L^{2}}=1$ and
\begin{align}\label{blowup:energy-afH-lower-bound}
\mathcal{E}^{\rm afH}_{N^{\nu},R_{N},\beta_N,g_N}[u_{N}] ={} & \ell_{N}^{-2}\left[\int_{\mathbb R^{2}}\left|(-{\rm i}\nabla + \beta_{N}\mathbf{A}_{R_{N}\ell_{N}^{-1}}[|w_{N}|^{2}])w_{N}\right|^{2} \right. \nonumber \\
& \left. - \frac{g_{N}}{2}\iint_{\mathbb R^{4}}W_{N^{\nu}\ell_{N}}|w_{N}(x)|^{2}|w_{N}(y)|^{2}{\rm d}x{\rm d}y\right]  + \ell_{N}^{s}\int_{\mathbb R^{2}}V|w_{N}|^{2}.
\end{align}

First, in the sub-critical regime where $g_{N} \nearrow g_{*}(0)$ faster than $\beta_{N}^{2} \to 0$, the magnetic potential $\mathbf{A}_{R_{N}\ell_{N}^{-1}}[|w_{N}|^{2}]$ in \eqref{blowup:energy-afH-lower-bound} have no contribution to the leading order term in the collapse regime. By using \eqref{ineq:gn}, \eqref{ineq:diamagnetic}, \eqref{cv-ineq:afH-afP-2body}, \eqref{blowup:energy-afH-lower-bound}, the proof of \eqref{blowup:minimizer-afH} is exactly the same as \eqref{blowup:minimizer-afP} in \Cref{sec:subcritical-afP}.

Next, in the critical regime where $g_{N} \equiv g_{*}(0)$, by using again \eqref{ineq:hardy}, \eqref{ineq:gn}, \eqref{ineq:diamagnetic}, \eqref{cv-ineq:afH-afP-2body}, \eqref{blowup:energy-afH-lower-bound} with $\mathbf{A}$ replaced by $\mathbf{A}_{R_{N}}$ and $\mathbf{A}_{R_{N}\ell_{N}^{-1}}$, the proof of \eqref{blowup:minimizer-afH} is similar to \eqref{blowup:minimizer-afP} in \Cref{sec:critical-afP}. The only thing we need to verify is the convergence \eqref{cv:minimizer-afP-nls-critical} where $\mathbf{A}$ replaced by $\mathbf{A}_{R_{N}\ell_{N}^{-1}}$, i.e.,
$$
\lim_{N \to \infty}\int_{\mathbb R^{2}}\mathbf{A}_{R_{N}\ell_{N}^{-1}}[|w_{N}|^{2}]|^{2}|w_{N}|^{2} = \int_{\mathbb R^{2}}\mathbf{A}[|w_{0}|^{2}]|^{2}|w_{0}|^{2}.
$$
This holds true, by \eqref{cv:afH-afP-curvature}, under the assumption \eqref{collapse:speed-hartree} since $R_{N}\ell_{N}^{-1} \sim N^{-\eta+\sigma} \to 0$ as $N\to\infty$.

Finally, in the super-critical regime where $g_{N} \searrow g_{*}(0)$, the compactness of the blowup sequence $\{w_{N}\}_{N}$ is not obtained directly by \eqref{ineq:gn-magnetic}. We recover this property under further assumption $\sigma < \frac{2(\eta-\nu)}{s+2}$ in addition to the assumption $\nu < \eta$ as seen in the proof of \eqref{cv:afH-afP-ground-state}. We assume on the contrary that $\|(-{\rm i}\nabla + \beta_{N}\mathbf{A}_{R_{N}\ell_{N}^{-1}}[|w_{N}|^{2}])w_{N}\|_{L^{2}} \to \infty$ as $N \to \infty$. We define $\widetilde{w_{N}} = \varepsilon_{N} w_{N}(\varepsilon_{N} \cdot)$ with 
\begin{equation}\label{boundedness:afH-blowup-length}
\varepsilon_{N}^{-2} := \int_{\mathbb R^{2}}\left|(-{\rm i}\nabla + \beta\mathbf{A}_{R_{N}\ell_{N}^{-1}}[|w_{N}|^{2}])w_{N}\right|^{2}.
\end{equation}
Noticing that
\begin{equation}\label{eq:afH-blowup-renormalized}
\int_{\mathbb R^{2}}\left|\widetilde{w_{N}}\right|^{2} = 1 = \int_{\mathbb R^{2}}\left|(-{\rm i}\nabla + \beta\mathbf{A}_{R_{N}\ell_{N}^{-1}\varepsilon_{N}^{-1}}[|\widetilde{w_{N}}|^{2}])\widetilde{w_{N}}\right|^{2}.
\end{equation}
Using \eqref{ineq:hardy}, \eqref{ineq:diamagnetic} and \eqref{eq:afH-ground-state-renormalized} with $\mathbf{A}$ replaced by $\mathbf{A}_{R_{N}\varepsilon_{N}^{-1}\varepsilon_{N}^{-1}}$, we deduce that $\widetilde{w_{N}}$ is bounded uniformly in $H^{1}(\mathbb R^{2})$. By using \eqref{blowup:energy-afH-lower-bound}, \eqref{cv:afH-afP-curvature}, \eqref{cv-ineq:afH-afP-2body}, \eqref{ineq:gn-magnetic} and dropping the nonnegative trapping potential, we obtain that
\begin{align*}
E^{\rm afH}_{N^{\nu},R_{N},\beta_N,g_N} \geq{} & \ell_{N}^{-2}\left[\int_{\mathbb R^{2}}\left|(-{\rm i}\nabla + \beta_{N}\mathbf{A}_{R_{N}\ell_{N}^{-1}\varepsilon_{N}^{-1}}[|\widetilde{w_{N}}|^{2}])\widetilde{w_{N}}\right|^{2} - \frac{g_{N}}{2}\int_{\mathbb R^{2}}|\widetilde{w_{N}}|^{4}\right] + \ell_{N}^{s}\int_{\mathbb R^{2}}V|w_{N}|^{2} \\
\geq{} & \ell_{N}^{-2}\varepsilon_{N}^{-2}\left[\int_{\mathbb R^{2}}\left|(-{\rm i}\nabla + \beta_{N}\mathbf{A}[|\widetilde{w_{N}}|^{2}])\widetilde{w_{N}}\right|^{2} - \frac{g_{N}}{2}\int_{\mathbb R^{2}}|\widetilde{w_{N}}|^{4} - C\beta_{N} R_{N}\ell_{N}^{-1}\varepsilon_{N}^{-1}\right] \\
\geq{} & \ell_{N}^{-2}\varepsilon_{N}^{-2}\left[\left(1-
\frac{g_{N}}{g_{*}(\beta_{N})}\right)\int_{\mathbb R^{2}}\left|(-{\rm i}\nabla + \beta_{N}\mathbf{A}[|\widetilde{w_{N}}|^{2}])\widetilde{w_{N}}\right|^{2} - C\beta_{N} R_{N}\ell_{N}^{-1}\varepsilon_{N}^{-1}\right] \\
\geq{} & \ell_{N}^{-2}\varepsilon_{N}^{-2}\left[\left(1-
\frac{g_{N}}{g_{*}(\beta)}\right)\int_{\mathbb R^{2}}\left|(-{\rm i}\nabla + \beta_{N}\mathbf{A}_{R_{N}\ell_{N}^{-1}\varepsilon_{N}^{-1}}[|\widetilde{w_{N}}|^{2}])\widetilde{w_{N}}\right|^{2} \right. \\
& \left. - C\left(1-
\frac{g_{N}}{g_{*}(\beta)}\right)\beta_{N} R_{N}\ell_{N}^{-1}\varepsilon_{N}^{-1} - C\beta_{N} R_{N}\ell_{N}^{-1}\varepsilon_{N}^{-1}\right].
\end{align*}
Recalling \eqref{eq:afH-blowup-renormalized}, the definition of the blowup length $\ell_{N}$ in \eqref{length:supercritical}, we deduce from the above that
\begin{equation}\label{boundedness:afH-blowup}
E^{\rm afH}_{N^{\nu},R_{N},\beta_N,g_N} \geq \varepsilon_{N}^{-2}\ell_{N}^{s}\left[\Gamma - C\beta_{N} R_{N}\ell_{N}^{-s-3}\varepsilon_{N}^{-1}\right].
\end{equation}
Here $\Gamma > 0$ is a universal positive constant which can be determined by \eqref{length:supercritical}. It is worth noting that, by \eqref{condition:potential-two-body} and \eqref{boundedness:afH-blowup-length},
$$
\ell_{N}^{-2}\varepsilon_{N}^{-2} = \int_{\mathbb R^{2}}\left|(-{\rm i}\nabla + \beta\mathbf{A}_{R_{N}}[|u_{N}|^{2}])u_{N}\right|^{2} \leq CN^{2\nu}.
$$
This implies that, under the assumption $\nu < \eta$ and $\sigma < \frac{2(\eta-\nu)}{s+2}$,
$$
\beta_{N} R_{N}\ell_{N}^{-s-3}\varepsilon_{N}^{-1} \leq CN^{\sigma\frac{s+2}{2}-\eta+\nu} \xrightarrow{N \to \infty} 0.
$$
Here we recalled again the definition of the blowup length $\ell_{N}$ in \eqref{length:supercritical}. Multiplying both side of \eqref{boundedness:afH-blowup} by $\varepsilon_{N}^{2}\ell_{N}^{-s}$, using the energy upper bound in \eqref{blowup:energy-afH}, and taking the limit $N\to\infty$ with notice that $\varepsilon_{N} \to 0$, we obtain a contradiction. Therefore, we must have that $\|(-{\rm i}\nabla + \beta_{N}\mathbf{A}_{R_{N}\ell_{N}^{-1}}[|w_{N}|^{2}])w_{N}\|_{L^{2}}$ is bounded uniformly. This in turn implies that $w_{N}$ is also bounded uniformly in $H^{1}(\mathbb R^{2})$, by using \eqref{ineq:hardy} and \eqref{ineq:diamagnetic} with $\mathbf{A}$ replaced by $\mathbf{A}_{R_{N}\ell_{N}^{-1}}$. We then use again \eqref{blowup:energy-afH-lower-bound}, \eqref{cv:afH-afP-curvature}, \eqref{cv-ineq:afH-afP-2body}, \eqref{ineq:gn-magnetic} and obtain that
$$
E^{\rm afH}_{N^{\nu},R_{N},\beta_N,g_N} \geq  \ell_{N}^{-2}\left[\int_{\mathbb R^{2}}\left|(-{\rm i}\nabla + \beta_{N}\mathbf{A}[|w_{N}|^{2}])w_{N}\right|^{2} - \frac{g_{N}}{2}\int_{\mathbb R^{2}}|w_{N}|^{4} - C\beta_{N} R_{N}\ell_{N}^{-1}\right] + \ell_{N}^{s}\int_{\mathbb R^{2}}V|w_{N}|^{2}.
$$
At this stage, it is worth noting that, by the definition of the blowup length $\ell_{N}$ in \eqref{length:supercritical} and the assumption \eqref{collapse:speed-hartree}, we have
$$
\beta_{N} R_{N}\ell_{N}^{-3} \sim N^{\sigma\frac{s+4}{2}-\eta}\ell_{N}^{s} = o(E^{\rm afH}_{N^{\nu},R_{N},\beta_N,g_N}).
$$
Then the rest of proof follows exactly the same as \eqref{blowup:minimizer-afP} in \Cref{sec:supercritical-afP} which is essentially based on the proof of it in the critical regime (see \Cref{sec:critical-afP}).
\end{proof}

\subsection{Condensation of the quantum energy and its ground states in the mean-field regime}

In the mean-field regime, we recall that $\beta>0$ and $0<g<g_{*}(\beta)$ are fixed. The rigorous derivation of the average-field-Pauli theory \eqref{functional:afP} in the mean-field limit regime can be obtained by combining the well-known results for non-interacting anyon gases \cite{LunRou-15} and for interacting Bose gases \cite{LewNamRou-16}. In the mentioned articles, the method of proof is the (quantitative) quantum de Finetti theory \cite{Rougerie-15,Rougerie-EMS}. The idea is to use the orthogonal projector 
$$
P := \mathbbm{1}(h\leq L) \quad \text{with} \quad h = -\Delta + V
$$
to project the many-body ground state into finite dimensional space and then bound the full (regularized) energy from below in terms of projected state of $H^{\rm af}_{N^{\nu},R,\beta,g}$. In particular, if $V(x) = |x|^{s}$ then the projected Hilbert space $PL^{2}(\mathbb R^{2})$ is of finite dimensional and we have (see, e.g., \cite[Lemma 3.3]{LewNamRou-16})
\begin{equation}\label{est:CLR}
\dim(PL^{2}(\mathbb R^{2})) \leq CL^{1+\frac{2}{s}} .
\end{equation}
Let $\Psi_{N}$ be a many-body ground state of $H^{\rm af}_{N^{\nu},R,\beta,g}$. It follows from \cite{LunRou-15,LewNamRou-16} that, with $P^{\perp} = \mathbbm{1} - P$ and for every $\varepsilon>0$,
\begin{align}\label{est:quantum-afH}
\frac{ \langle \Psi_{N}, H^{\rm af}_{N^{\nu},R,\beta,g} \Psi_{N} \rangle }{N} \geq{} & \int_{SPL^{2}(\mathbb R^{2})} \mathcal{E}^{\rm afH}_{N^{\nu},R,\beta,g}[u] {\rm d}\mu_{\Psi_{N}}^{(3)}(u)  + CL\tr\left[P^{\perp}\gamma_{\Psi_{N}}^{(1)}\right] \\
& - C\frac{L^{2+\frac{2}{s}}}{N} \left(1+|\log R|\right) - \frac{C_{\beta}}{N} - \left(\frac{C_{\varepsilon}}{L^{\frac{1}{2}}R^{1+\varepsilon}} + \frac{C}{LR^{2}} + \frac{CN^{\frac{3\nu}{2}}}{L^{\frac{1}{2}}}\right) \tr\left[h\gamma_{\Psi_{N}}^{{(1)}}\right]. \nonumber
\end{align}
Here it is worth noting that the last error term in the above came from the estimate of the two-body interactions when projecting onto the finite dimensional space. For this purpose, we made use of \cite[Equation (44)]{LewNamRou-17} as well as the facts that $W_{N^{\nu}} \leq CN^{2\nu}$ and that $Q \leq L^{-1}h$.  We do not detail the proof of this since a similar argument has been detailed in \cite{NguRic-24}. It is worth noting that the authors in \cite{LewNamRou-17} used a refined operator bound for two-body interactions and a second moment estimate in order to obtain BECs for a wide range of $\nu>0$ (see also \cite{NamRou-20}). This, however, is not available for three-body problem in high dimensional space such as interacting anyon gases in the present article.

We next optimize the error in \eqref{est:quantum-afH} by choosing optimally $L>0$. Assuming that $R>0$ behaves at worst as $R \sim N^{-\eta}$ with $\eta>0$. We might ignore the $|\log R|$ and $R^{\varepsilon}$ factors, by changing a little bit $\eta$. Roughly speaking, by a bootstrap argument as used in \cite{LewNamRou-17,NamRou-20}, we can prove that $\tr\left[h\gamma_{\Psi_{N}}^{{(1)}}\right]$ is bounded uniformly, provided that
$$
0 < \nu < \frac{s}{6s+6} \quad \text{and} \quad 0 < \eta <\frac{s}{4s+4}.
$$
To see this, we consider the modified Hamiltonian by a small perturbation of non-interacting Hamiltonian, i.e., with $0<\varepsilon<1$,
$$
H^{\rm af}_{N^{\nu},R,\beta,g,\varepsilon} := H^{\rm af}_{N^{\nu},R,\beta,g} - \varepsilon H^{\rm af}_{N^{\nu},R,\beta,0}.
$$
Let $E^{\rm afQM}_{N^{\nu},R,\beta,g,\varepsilon}$ be the corresponding ground state energy. We then make use of the Schr{\"o}\-dinger equation satisfying by the ground state $\Psi_{N}$ of $H^{\rm af}_{N^{\nu},R,\beta,g}$ and deduce that
\begin{equation}\label{boundedness:qm-ground-state}
\frac{\langle \Psi_{N} | H^{\rm af}_{N^{\nu},R,\beta,0} | \Psi_{N}\rangle}{N} \leq \frac{E^{\rm afQM}_{N^{\nu},R,\beta,g} - E^{\rm afQM}_{N^{\nu},R,\beta,g,\varepsilon}}{\varepsilon} \leq C\frac{1+\left|E^{\rm afQM}_{N^{\nu},R,\beta,g,\varepsilon}\right|}{\varepsilon}.
\end{equation}
The uniform boundedness of the above follows from \eqref{est:quantum-afH} and a bootstrap argument as in \cite{LewNamRou-17,NamRou-20}. This also yields that $\tr\left[h\gamma_{\Psi_{N}}^{{(1)}}\right]$ is bounded uniformly as well (see \cite[Proposition 2.6]{LunRou-15}). Looking back at \eqref{est:quantum-afH}, it is necessary that $L \gg N^{3\nu}$ and $L \gg N^{2\eta}$ in order to obtain the convergence of energy.  We thus have to optimize
$$
\frac{L^{2+\frac{2}{s}}}{N} + \frac{N^{\frac{3\nu}{2}} + N^{\eta}}{L^{\frac{1}{2}}}.
$$
The optimal choice of $L>0$ is
\begin{equation}\label{eq:L}
L \sim \left(N^{\frac{3\nu}{2}+1} + N^{\eta+1} \right)^{\frac{2s}{5s+4}} \quad \text{with} \quad \nu < \frac{s}{6s+6} \quad \text{and} \quad 0 < \eta <\frac{s}{4s+4},
\end{equation}
which yields that the de Finetti measure $\mu_{\Psi_{N}}^{(3)}$ is tight as $N \to \infty$. Indeed, this follows from the above choice of $L$ and the estimates
\begin{equation}\label{est:de-Finetti-measure}
1 \geq \int_{SPL^{2}(\mathbb R^{2})} {\rm d} \mu_{\Psi_{N}}^{(3)}(u) = \tr\left[P^{\otimes 3}\gamma_{\Psi_{N}}^{(3)} P^{\otimes 3}\right] \geq 1- 3\tr\left[P^{\perp}\gamma_{\Psi_{N}}^{(1)}\right] \geq 1 - 3L^{-1} \tr\left[h\gamma_{\Psi_{N}}^{(1)}\right].
\end{equation}
Looking back at \eqref{est:quantum-afH}, by collecting all above estimates, we arrive at the final estimate
\begin{align}
E^{\rm afH}_{N^{\nu},R_{N},\beta,g} \geq E^{\rm afQM}_{N^{\nu},R_{N},\beta,g} & \geq \int_{SPL^{2}(\mathbb R^{2})} \mathcal{E}^{\rm afH}_{N^{\nu},R_{N},\beta,g}[u] {\rm d}\mu_{\Psi_{N}}^{(3)}(u) - CN^{\frac{\nu(6s+6)-s}{5s+4}} - CN^{\frac{\eta(4s+4)-s}{5s+4}} \label{cv:ground-state-qm-afH} \\
& \geq E^{\rm afH}_{N^{\nu},R_{N},\beta,g} \int_{SPL^{2}(\mathbb R^{2})}{\rm d}\mu_{\Psi_{N}}^{(3)}(u) - CN^{\frac{\nu(6s+6)-s}{5s+4}} - CN^{\frac{\eta(4s+4)-s}{5s+4}}.\label{cv:energy-qm-afH}
\end{align}
Here we have used the fact that, for $N$ large enough, $\mathcal{E}^{\rm afH}_{N^{\nu},R_{N},\beta,g}[u]\geq 0$ (see \Cref{thm:collapse-afH}). It is worth noting that it was required $\nu<\eta$ in the case where $g_{*}(0) \leq g < g_{*}(\beta)$. In the stable regime where $0 \leq g < g_{*}(\beta)$, we deduce the desired convergence \eqref{cv:qm-afP-energy} of the average-field quantum energy to the average-field-Pauli energy from \eqref{cv:energy-qm-afH}, \eqref{est:de-Finetti-measure} and \eqref{cv:afH-afP-energy}. 

Next, we prove the convergence of many-body ground states in \eqref{cv:qm-afP-energy}. Let $\Psi_{N}$ be a sequence of ground states of $H^{\rm af}_{N^{\nu},R,\beta,g}$. Since $\tr\left[h\gamma_{\Psi_{N}}^{(1)}\right]$ is bounded uniformly, modulo a subsequence, we may assume that (see \cite{LewNamRou-14})
$$
\lim_{N \to \infty}\tr\left| \gamma_{\Psi_{N}}^{(k)} - \gamma^{(k)}\right| = 0,\quad \forall k\in\mathbb N.
$$
On the other hand, by the quantitative quantum de Finetti theorem (see e.g., \cite{LewNamRou-16}) and \eqref{est:CLR}, \eqref{eq:L}, we have
\begin{equation}\label{boundedness:qm-ground-state-finite}
\tr\left|P^{\otimes 3}\gamma_{\Psi_{N}}^{(3)}P^{\otimes 3} - \int_{SPL^{2}(\mathbb R^{2})}|u^{\otimes 3}\rangle\langle u^{\otimes 3}|{\rm d}\mu_{\Psi_{N}}^{(3)}(u)\right| \leq C\frac{\dim(PL^{2}(\mathbb R^{2}))}{N} \leq C\frac{L^{1+\frac{2}{s}}}{N} \xrightarrow{N\to\infty} 0.
\end{equation}
It is worth noting that $\mu_{\Psi_{N}}^{(3)}(SPL^{2}(\mathbb R^{2})) \to 1$ as $N\to\infty$, by \eqref{est:de-Finetti-measure}. Then it follows from \eqref{boundedness:qm-ground-state-finite} that
\begin{equation}\label{boundedness:qm-ground-state-Finetti}
\lim_{N\to\infty}\tr\left|\gamma_{\Psi_{N}}^{(3)} - \int_{SPL^{2}(\mathbb R^{2})}|u^{\otimes 3}\rangle\langle u^{\otimes 3}|{\rm d}\mu_{\Psi_{N}}^{(3)}(u)\right| = 0.
\end{equation}
Testing the above with a sequence of finite rank orthogonal projectors $P_{K} \to 0$ as $K\to\infty$, and using the strong convergence of $\gamma_{\Psi_{N}}^{(3)}$, we obtain that the measure $\{\mu_{\Psi_{N}}^{(3)}\}_{N}$ is tight on the set of one-body pure states. Therefore, there exists a limit measure $\mu$ supported on the unit ball of $L^{2}(\mathbb R^{2})$. We then deduce that
$$
\lim_{N\to\infty}\tr\left|\gamma_{\Psi_{N}}^{(3)} - \int_{SL^{2}(\mathbb R^{2})}|u^{\otimes 3}\rangle\langle u^{\otimes 3}|{\rm d}\mu(u)\right| = 0
$$
which in turn implies that
\begin{equation}\label{cv:quantum-ground state}
\lim_{N\to\infty}\tr\left|\gamma_{\Psi_{N}}^{(1)} - \int_{SL^{2}(\mathbb R^{2})}|u \rangle\langle u|{\rm d}\mu(u)\right| = 0.
\end{equation}
To complete the proof of \eqref{cv:qm-afP-ground-state}, it remains to prove that $\mu$ in \eqref{cv:quantum-ground state} is supported on the set of average-field-Pauli ground states $\mathcal{M}^{\rm afP}$.

Looking back at \eqref{cv:ground-state-qm-afH}, we split the integral into two disjoint parts, i.e., $SPL^{2}(\mathbb R^{2}) = \mathrm{K}_{-} \cup \mathrm{K}_{+}$ where
$$
\mathrm{K}_{-} = \{u\in L^{2}(\mathbb R^{2}) : \mathcal{E}^{\rm afP}_{R_{N},\beta,g}[u] \leq C_{\mathrm{kin}}\} \quad \text{and} \quad \mathrm{K}_{+} = L^{2}(\mathbb R^{2}) \setminus \mathrm{K}_{-}.
$$
with $C_{\mathrm{kin}} > 0$ a large constant independent of $N$. For the high kinetic energy part $\mathrm{K}_{+}$, we consider the perturbed average-field-Hartree functional
$$
\mathcal{E}^{\rm afH}_{N^{\nu},R_{N},\beta,g,\varepsilon}[u] = \mathcal{E}^{\rm afH}_{N^{\nu},R_{N},\beta,g}[u] - \varepsilon \mathcal{E}^{\rm afP}_{R_{N},\beta,0}[u]
$$
with the corresponding modified Hartree energy $E^{\rm afH}_{N^{\nu},R_{N},\beta,g,\varepsilon}$. Here $\lambda$ is fixed such that either $0<\varepsilon < 1-\frac{g}{g_{*}(0)}$ if $g<g_{*}(0)$ or $0<\varepsilon<1$ if $g\geq g_{*}(0)$ and $\nu < \eta$. We may return to \Cref{thm:collapse-afH} and derive a similar energy lower bound \eqref{cv:afH-afP-energy} with $E^{\rm afH}_{N^{\nu},R_{N},\beta,g}$ replaced by $E^{\rm afH}_{N^{\nu},R_{N},\beta,g,\varepsilon}$. In particular, we find that $E^{\rm afH}_{N^{\nu},R_{N},\beta,g,\varepsilon} \geq -C_{\varepsilon}$ and deduce that
\begin{equation}\label{cv:ground-states-split-1}
\mathcal{E}^{\rm afH}_{N^{\nu},R_{N},\beta,g}[u] \geq \varepsilon \mathcal{E}^{\rm afP}_{R_{N},\beta,0}[u] - C_{\varepsilon} \geq \frac{\varepsilon C_{\mathrm{kin}}}{2},\quad\forall u \in \mathrm{K}_{+}
\end{equation}
for a large enough $C_{\rm kin}>0$. Plugin \eqref{cv:ground-states-split-1} into \eqref{cv:ground-state-qm-afH}, we deduce that
\begin{align*}
E^{\rm afH}_{N^{\nu},R_{N},\beta,g} \geq E^{\rm afQM}_{N^{\nu},R_{N},\beta,g} \geq{} & \int_{\mathrm{K}_{+}}\frac{\varepsilon C_{\mathrm{Kin}}}{2}{\rm d}\mu_{\Psi_{N}}^{(3)}(u) + \int_{\mathrm{K}_{-}} \mathcal{E}^{\rm afH}_{N^{\nu},R_{N},\beta,g}[u]{\rm d}\mu_{\Psi_{N}}^{(3)}(u) - o(1)_{N\to\infty} \\
\geq{} & \int_{SPL^{2}(\mathbb R^{2})}\min\left\{\frac{\varepsilon C_{\mathrm{Kin}}}{2},\mathcal{E}^{\rm afH}_{N^{\nu},R_{N},\beta,g}[u]\right\}{\rm d}\mu_{\Psi_{N}}^{(3)}(u) - o(1)_{N\to\infty}.
\end{align*}
Passing to the limit $N\to\infty$, using the convergences of energy \eqref{cv:afH-afP-energy}, of ground states \eqref{cv:afH-afP-ground-state} and of measures $\mu_{\Psi_{N}}^{(3)} \to \mu$, and taking finally $C_{\rm kin} \to \infty$, we arrive at
\begin{equation}\label{cv:ground-states-last-step}
E^{\rm afP}_{\beta,g} \geq \lim_{N\to \infty} E^{\rm afQM}_{N^{\nu},R_{N},\beta,g} \geq \int_{SL^{2}(\mathbb R^{2})} E^{\rm afP}_{\beta,g}[u] {\rm d}\mu (u) \geq E^{\rm afP}_{\beta,g}.
\end{equation}
The above shows that $\mu$ must be supported on $\mathcal{M}^{\rm afP}$. This proved \eqref{cv:qm-afP-ground-state}.

\subsection{Collapse of the quantum energy and its ground states in the weak field regime}

We complete the proof of \Cref{thm:many-body} by establishing the collapse phenomenon. The collapse of average-field quantum energy in \eqref{blowup:qm-energy} follows from \eqref{cv:energy-qm-afH}, \eqref{est:de-Finetti-measure} and \Cref{thm:collapse-afH}. It remains to prove the collapse of average-field many-body ground states in \eqref{blowup:qm-ground-state}.

Looking back at \eqref{boundedness:qm-ground-state}, we choose $\varepsilon = \varepsilon_{N}$ equaling to either $C\left(1-\frac{g_{N}}{g_{*}(0)}\right)$ in the sub-critical case, or $C\beta_{N}^{2}$ in the critical and super-critical cases. Here $C>0$ is a universal small constant. We may return to \Cref{thm:many-body} and derive a similar asymptotic formula for $E^{\rm afQM}_{N^{\nu},R,\beta,g,\varepsilon_{N}}$. We then deduce from \eqref{boundedness:qm-ground-state} and arguments in \cite[Proposition 2.6]{LunRou-15} that
$$
\tr\left[h\gamma_{\Psi_{N}}^{{(1)}}\right] \leq C\ell_{N}^{-2}.
$$
Plugin this into \eqref{est:quantum-afH} and \eqref{est:de-Finetti-measure}, we obtain the estimate of energy
\begin{align}\label{est:quantum-afH-collapse}
E^{\rm afH}_{N^{\nu},R_{N},\beta_{N},g_{N}} & \geq \frac{ \langle \Psi_{N}, H^{\rm af}_{N^{\nu},R_{N},\beta_{N},g_{N}} \Psi_{N} \rangle }{N} \\
& \geq \int_{SPL^{2}(\mathbb R^{2})} \mathcal{E}^{\rm afH}_{N^{\nu},R_{N},\beta_{N},g_{N}}[u] {\rm d}\mu_{\Psi_{N}}^{(3)}(u) - C\frac{L^{2+\frac{2}{s}}}{N} - C\frac{N^{\frac{3\nu}{2}} + N^{\eta}}{L^{\frac{1}{2}}} \ell_{N}^{-2} \nonumber
\end{align}
and the estimate of the de Finetti measure
\begin{equation}\label{est:de-Finetti-measure-collapse}
1 \geq \int_{SPL^{2}(\mathbb R^{2})} {\rm d} \mu_{\Psi_{N}}^{(3)}(u) \geq 1 - 3L^{-1}\ell_{N}^{-2}.
\end{equation}
It is straightforward that, if additionally to \eqref{blowup:speed-many-body} we assume further that 
$$
\sigma < \min\left\{\frac{s-\nu(6s+6)}{5s^{2}+12s+8},\frac{s-\eta(4s+4)}{5s^{2}+12s+8}\right\}
$$
then we can choose $L > 0$ optimally in \eqref{est:quantum-afH-collapse} such that
\begin{equation}\label{cv:qm-energy-measure}
\lim_{N \to \infty} \int_{SPL^{2}(\mathbb R^{2})} \frac{\mathcal{E}^{\rm afH}_{N^{\nu},R_{N},\beta_{N},g_{N}}[u]}{E^{\rm afH}_{N^{\nu},R_{N},\beta_{N},g_{N}}}{\rm d}\mu_{\Psi_{N}}^{(3)}(u)  = 1 = \lim_{N \to \infty} \int_{SPL^{2}(\mathbb R^{2})} {\rm d}\mu_{\Psi_{N}}^{(3)}(u).
\end{equation}
Here the second equality in \eqref{cv:qm-energy-measure} follows from \eqref{est:de-Finetti-measure-collapse}, \eqref{cv:speed-many-body} and \eqref{blowup:speed-many-body}. 

Next, by using \eqref{cv:qm-energy-measure} and \eqref{boundedness:qm-ground-state-finite}, we then obtain again \eqref{boundedness:qm-ground-state-Finetti} which in turn implies that
$$
\lim_{N \to \infty} \tr\left|\gamma_{\Psi_{N}}^{(1)} - \int_{SPL^{2}(\mathbb R^{2})} |u\rangle \langle u| {\rm d}\mu_{\Psi_{N}}^{(3)}(u)\right| = 0.
$$
The convergence of the one-particle density matrix then reduces to the following
\begin{equation}\label{blowup:qm-ground-state-completed}
\lim_{N \to \infty} \int_{SPL^{2}(\mathbb R^{2})} \left|\langle u, Q_{N} \rangle\right| {\rm d}\mu_{\Psi_{N}}^{(3)}(u) = 1 ,
\end{equation}
where $Q_{N} = \ell_{N}^{-1}Q_{0}(\ell_{N}^{-1}\cdot)$. To prove this, we come back to \eqref{cv:qm-energy-measure}. We split the integral into two disjoint parts: one is of small de Finetti measure and one contains ``approximate'' average-field-Hartree ground states. For this purpose, we define the nonnegative sequence of numbers
\begin{equation}\label{number:approximate-afP}
\delta_{N} := \int_{SPL^{2}(\mathbb R^{2})} \left(\frac{\mathcal{E}^{\rm afH}_{N^{\nu},R_{N},\beta_{N},g_{N}}[u]}{E^{\rm afH}_{N^{\nu},R_{N},\beta_{N},g_{N}}} - 1 \right) {\rm d}\mu_{\Psi_{N}}^{(3)}(u)
\end{equation}
which converges to $0$ as $N\to\infty$, by \eqref{cv:qm-energy-measure}. Furthermore, we introduce the set
\begin{equation}\label{set:approximate-afP}
\mathcal{T}_{N} := \left\{u \in SPL^{2}(\mathbb R^{2}) : 0\leq \frac{\mathcal{E}^{\rm afH}_{N^{\nu},R_{N},\beta_{N},g_{N}}[u]}{E^{\rm afH}_{N^{\nu},R_{N},\beta_{N},g_{N}}} -1 \leq \sqrt{\delta_{N}}, \int_{\mathbb R^{2}}|u|^{2} = 1\right\}.
\end{equation}
For any sequence $\{u_{N}\}_{N} \subset \mathcal{T}_{N}$, by the same arguments as in the proof of \Cref{thm:collapse-afH}, we can prove that $\ell_{N}u_{N}(\ell_{N}\cdot)$ converges strongly to $Q_{0}$ in $L^{2}(\mathbb R^{2})$. Equivalently, we have
$$
\lim_{N \to \infty} \left|\langle u_{N}, Q_{N} \rangle\right| = 1
$$
which in turn yields that
\begin{equation}\label{cv:approximate-afP}
\lim_{N \to \infty}\inf_{u\in \mathcal{T}_{N}} \left|\langle u, Q_{N} \rangle\right| = 1 .
\end{equation}
On the other hand, by the definition of $\delta_{N}$ in \eqref{number:approximate-afP} and of $\mathcal{T}_{N}$ in \eqref{set:approximate-afP}, we have
$$
\delta_{N} \geq \int_{\mathcal{T}_{N}'} \left( \frac{ \mathcal{E}^{\rm afH}_{N^{\nu},R_{N},\beta_{N},g_{N}}[u] }{E^{\rm afH}_{N^{\nu},R_{N},\beta_{N},g_{N}}} - 1 \right) {\rm d}\mu_{\Psi_{N}}^{(3)}(u) \geq \sqrt{\delta_{N}} \mu_{\Psi_{N}}^{(3)}(\mathcal{T}_{N}').
$$
Since $\delta_{N} \to 0$ as $N \to \infty$, the above yields that $\mu_{\Psi_{N}}^{(3)}(\mathcal{T}_{N}') \to 0$ and hence $\mu_{\Psi_{N}}^{(3)}(\mathcal{T}_{N}) \to 1$, by \eqref{cv:qm-energy-measure}. We then deduce from \eqref{cv:approximate-afP} that
$$
\int_{SPL^{2}(\mathbb R^{2})} \left|\langle u, Q_{N} \rangle\right| {\rm d}\mu_{\Psi_{N}}^{(3)}(u) \geq \int_{\mathcal{T}_{N}} \left|\langle u, Q_{N} \rangle\right| {\rm d}\mu_{\Psi_{N}}^{(3)}(u) \geq \mu_{\Psi_{N}}^{(3)}(\mathcal{T}_{N})\inf_{u \in \mathcal{T}_{N}} \left|\langle u, Q_{N} \rangle\right| \xrightarrow{N\to\infty} 1.
$$
This is the desired convergence \eqref{blowup:qm-ground-state-completed}. The proof of \eqref{blowup:qm-ground-state} is finished.

%

\appendix

\section{Magnetic Gagliardo--Nirenberg optimizers in the weak-field limit}\label{app:gn}

\begin{proposition}
Let $u_{0}$ be the (unique) minimizer for $g_{*}(0)$ and $u_{\beta}$ be a sequence of minimizers of $g_{*}(\beta)$, given by \eqref{eq:critical-value}, in the limit $\beta \to 0$. Then we have, for the whole sequence, $u_{\beta} \to u_{0}$ strongly in $H^{1}(\mathbb R^{2})$, up to translation, dilation and phase transition. Furthermore, $g_{*}(\beta) \to g_{*}(0)$.
\end{proposition}

\begin{proof}
The convergence $g_{*}(\beta) \to g_{*}(0)$ follows immediately from the estimates \eqref{ineq:gn-weak-field}. On the other hand, the proof of convergence of minimizers is based on concentration compactness method \cite{Lions-84a,Lions-84b}. By scaling, we assume without loss of generality that $\|u_{\beta}\|_{L^{4}}^{4} = 1$. Then $\{u_{\beta}\}$ is bounded uniformly in $H^{1}(\mathbb R^{2})$. Hence, up to translation, dilation and phase transition, $u_{\beta}$ converges to $u_{0}$ weakly in $H^{1}(\mathbb R^{2})$ and almost everywhere in $\mathbb R^{2}$. If $\|u_{0}\|_{L^{2}}^{2} = 1$ then, by Brezis--Lieb lemma \cite{BreLie-83}, $u_{\beta}$ converges to $u_{0}$ strongly in $L^{2}(\mathbb R^{2})$. In fact, this strong convergence holds true in $L^{p}(\mathbb R^{2})$ for every $2\leq p < \infty$. This yields that, by Fatou's lemma (see \eqref{cv:minimizer-afP-weak-lower-semicontinuous}),
$$
g_{*}(0) = \lim_{\beta \to 0}g_{*}(\beta) = \lim_{\beta \to 0} \frac{\int_{\mathbb R^{2}} \left|(-{\rm i}\nabla + \beta\mathbf{A}[|u_{\beta}|^{2}]) u_{\beta}\right|^{2}}{\frac{1}{2}\int_{\mathbb R^{2}} |u_{\beta}|^4} \geq \frac{\int_{\mathbb R^{2}} |\nabla u_{0}|^{2}}{\frac{1}{2}\int_{\mathbb R^{2}} |u_{0}|^4}.
$$
This proves that $u_{0}$ is the minimizer for $g_{*}(0)$, by \eqref{ineq:gn}. Since this minimizer $u_{0}$ is unique (up to translation and dilation), the convergence holds for the whole sequence.

It remains to eliminate the case where $0 \leq \|u_{0}\|_{L^{2}}^{2} < 1$. Indeed, if $\|u_{0}\|_{L^{2}}^{2} = 0$ then we must have that $u_{\beta} \to 0$ strongly in $L^{p}(\mathbb R^{2})$ for every $2 < p <\infty$ (see \cite{Lions-84b}). However, this contradicts our assumption that $\|u_{\beta}\|_{L^{4}}^{4} = 1$. Finally, we now assume that $0 < \|u_{0}\|_{L^{2}}^{2} < 1$. Then it follows from \eqref{ineq:hardy} and Brezis--Lieb lemma \cite{BreLie-83} that
\begin{align*}
g_{*}(0) = \lim_{\beta \to 0}g_{*}(\beta) & = \lim_{\beta \to 0} \frac{\int_{\mathbb R^{2}} \left|(-{\rm i}\nabla + \beta\mathbf{A}[|u_{\beta}|^{2}]) u_{\beta}\right|^{2}}{\frac{1}{2}\int_{\mathbb R^{2}} |u_{\beta}|^4} \\
& \geq \lim_{\beta \to 0} \frac{\int_{\mathbb R^{2}} |\nabla u_{\beta}|^{2}}{\frac{1}{2}\int_{\mathbb R^{2}} |u_{\beta}|^4} \\
& = \liminf_{\beta \to 0} \frac{\int_{\mathbb R^{2}} |\nabla u_{0}|^{2} + \int_{\mathbb R^{2}} |\nabla (u_{\beta}-u_{0})|^{2}}{\frac{1}{2}\int_{\mathbb R^{2}} |u_{0}|^4 + |u_{\beta}-u_{0}|^4} \\
& \geq \min \left\{\frac{\int_{\mathbb R^{2}} |\nabla u_{0}|^{2}}{\int_{\mathbb R^{2}} |u_{0}|^4}, \liminf_{\beta \to 0} \frac{\int_{\mathbb R^{2}} |\nabla (u_{\beta}-u_{0})|^{2}}{\frac{1}{2}\int_{\mathbb R^{2}} |u_{\beta}-u_{0}|^4}\right\}.
\end{align*}
However, this is impossible. Indeed, by \eqref{ineq:gn} and \eqref{ineq:diamagnetic}, we have
$$
\frac{\int_{\mathbb R^{2}} |\nabla u_{0}|^{2}}{\frac{1}{2}\int_{\mathbb R^{2}} |u_{0}|^4} \geq \frac{g_{*}(0)}{\int_{\mathbb R^{2}} |u_{0}|^{2}} > g_{*}(0).
$$
Furthermore, by again \eqref{ineq:gn} and \eqref{ineq:diamagnetic},
$$
\liminf_{\beta \to 0} \frac{\int_{\mathbb R^{2}} |\nabla (u_{\beta}-u_{0})|^{2}}{\frac{1}{2}\int_{\mathbb R^{2}} |u_{\beta}-u_{0}|^4} \geq \liminf_{\beta \to 0} \frac{g_{*}(0)}{\int_{\mathbb R^{2}} |u_{\beta}-u_{0}|^{2}} = \frac{g_{*}(0)}{1-\int_{\mathbb R^{2}} |u_{0}|^{2}} > g_{*}(0).
$$
Collecting all the above estimates, we obtain a contradiction. The proof is completed.
\end{proof}


\end{document}